\documentclass[11pt]{article}
\usepackage{amsfonts,amsthm,amssymb}
\usepackage[fleqn]{amsmath}
\usepackage[pdftex]{graphicx}
\usepackage{color}
\newcommand{\be}{\begin{eqnarray}}
\newcommand{\ee}{\end{eqnarray}}
\newcommand{\bez}{\begin{eqnarray*}}
\newcommand{\eez}{\end{eqnarray*}}

\theoremstyle{plain}
\newtheorem{theorem}{Theorem}[section]
\newtheorem{lemma}[theorem]{Lemma}
\newtheorem{proposition}[theorem]{Proposition}

\theoremstyle{definition}

\newtheorem{remark}[theorem]{Remark}
\newtheorem{example}[theorem]{Example}

\numberwithin{equation}{section}
\numberwithin{theorem}{section}

\allowdisplaybreaks

\begin{document}

\title{\bf Matrix KP: tropical limit and Yang-Baxter maps}

\author{
\sc{Aristophanes Dimakis}$^a$ and \sc{Folkert M\"uller-Hoissen}$^b$ \\
  \small
 $^a$ Dept. of Financial and Management Engineering,
 University of the Aegean, Chios, Greece \\
 \small E-mail: dimakis@aegean.gr  \\
  \small
 $^b$ Max-Planck-Institute for Dynamics and Self-Organization,
         G\"ottingen, Germany \\
 \small E-mail: folkert.mueller-hoissen@ds.mpg.de        
}

\date{}

\maketitle

\begin{abstract}
We study soliton solutions of matrix Kadomtsev-Petviashvili (KP) equations in a tropical limit, in which their support 
at fixed time is a planar graph and polarizations are attached to its constituting lines.  
There is a subclass of ``pure line soliton solutions'' for which we find that, in this limit, 
the distribution of polarizations is fully determined by a Yang-Baxter map. For a \emph{vector} 
KP equation, this map is given by an $R$-matrix, whereas it is a non-linear map in case of a 
more general matrix KP equation. We also consider the corresponding Korteweg-deVries (KdV) reduction.
Furthermore, exploiting the fine structure of soliton interactions in the 
tropical limit, we obtain a new solution of the tetrahedron (or Zamolodchikov) equation. 
Moreover, a solution of the functional tetrahedron equation arises from the parameter-dependence 
of the vector KP $R$-matrix. 
\end{abstract}

\section{Introduction}
\label{sec:intro}
A line soliton solution of the scalar Kadomtsev-Petviashvili (KP-II) equation (see, e.g., \cite{Koda17}) is, 
at fixed time $t$, an exponentially localized wave on a plane. 
The ``tropical limit'', in the sense of our work in 
\cite{DMH11KPT,DMH12KPBT,DMH14KdV} (also see \cite{Bion+Chak06,Chak+Koda08JPA}), 
takes it to a piecewise linear structure, a planar graph that represents the wave crest, with values 
of the dependent variable attached to its edges. 

In this work we consider the $m \times n$ \emph{matrix} potential KP equation 
\be
     4 \phi_{xt} - \phi_{xxxx} - 3 \phi_{yy} - 6 (\phi_x K \phi_x)_x + 6 (\phi_x K \phi_y - \phi_y K \phi_x) = 0 \, , 
                    \label{KmatrixpKP}
\ee
where $K$ is a constant $n \times m$ matrix and $\phi$ an $m \times n$ matrix, depending on independent variables $x,y,t$, 
and a subscript indicates a corresponding partial derivative. We will refer to this equation as pKP$_K$. 

If $\phi$ is a solution of (\ref{KmatrixpKP}), then $\phi_R := \phi K$ and $\phi_L := K \phi$ solve the ordinary 
$m \times m$, respectively $n \times n$, matrix potential KP equation. 
We also note that, if $K = T K' S$ with a constant $m' \times m$ matrix $S$ and a constant  
$n \times n'$ matrix $T$, then the $m' \times n'$ matrix $\phi' = S \phi T$ satisfies the pKP$_{K'}$ equation, 
as a consequence of (\ref{KmatrixpKP}).

In the vector case $n=1$, writing $K = (k_1,\ldots,k_m)$ and $\phi = (\phi_1, \ldots,\phi_m)^\intercal$, 
(\ref{KmatrixpKP}) becomes the following system of coupled equations,
\bez
    4 \, \phi_{i,xt} - \phi_{i,xxxx} - 3 \, \phi_{i,yy} - 6 \sum_{j=1}^m k_j \, \Big( 
        (\phi_{i,x} \phi_{j,x})_x - \phi_{i,x} \phi_{j,y} + \phi_{i,y} \phi_{j,x} \Big) = 0 \qquad
        i=1,\ldots,m \, .
\eez
By choosing $T=1$ and any invertible $m \times m$ matrix $S$ that has $K$ as 
its first row, we have $K = K' S$ with $K' = (1,0,\ldots,0)$. In terms of the new variable $\phi' = S \phi$, 
the above system thus consists of one scalar pKP equation and $m-1$ linear equations involving the 
dependent variable of the former. 

 For 
\bez
     u := 2 \, \phi_x \,  ,
\eez
we obtain from (\ref{KmatrixpKP}) the $m \times n$ matrix KP equation
\be
     ( \, 4 \, u_t - u_{xxx} - 3 \, (u K u)_x \, )_x - 3 \, u_{yy} + 3 \, \Big( 
     u K \int u_y \, dx - \int u_y \, dx \, K u \Big)_x = 0 \, .
                    \label{KmatrixKP}
\ee
The extension of the scalar KP equation to a matrix version achieves that solitons carry internal degrees of 
freedom, ``polarization''. 

The Korteweg-deVries (KdV) reduction of (\ref{KmatrixKP}) is
\be
   4 \, u_t - u_{xxx} - 3 \, (u K u)_x = 0 \, ,
\ee
which we will refer to as KdV$_K$. If $K$ is the identity matrix, this is the matrix KdV equation (see, e.g., \cite{Gonc01}). 
The 2-soliton solution of the latter yields a map from polarizations at $t \ll 0$ to polarizations at $t \gg 0$. 
It is known \cite{Vese03,Gonc+Vese04} that this yields a Yang-Baxter map, i.e., a set-theoretical solution 
of the (quantum) Yang-Baxter equation (also see \cite{Tsuc04} for the case of the vector Nonlinear Schr\"odinger equation). 
Not surprisingly, this is a feature preserved in the tropical limit. The surprising new insight, however, is that 
this map governs the evolution of polarizations throughout the tropical limit graph of a soliton solution. 
In case of a \emph{vector} KdV equation, i.e., KdV$_K$ with $n=1$, it is given by an $R$-matrix, a \emph{linear} 
map solution of the Yang-Baxter equation.

More generally, we will explore in this work the tropical limit of ``pure'' (see Section~\ref{sec:pure}) 
soliton solutions of the above $K$-modified matrix KP equation and demonstrate that a Yang-Baxter map 
governs their structure. 

In case of the vector KP equation, the expression for a pure soliton solution 
involves a function $\tau$ which is a $\tau$-function of the scalar KP equation. Its tropical limit 
at fixed $t$ determines a planar graph, and the vector KP soliton solution associates in this limit 
a constant vector (polarization) with each linear segment of the graph. The polarization values are 
then related by a linear Yang-Baxter map, represented by an $R$-matrix, which does not depend on the 
independent variables $x,y,t$, but only on the ``spectral parameters'' of the soliton solution.  

Section~\ref{sec:sol} summarizes a binary Darboux transformation for the pKP$_K$ equation and applies 
it to a trivial seed solution in order to obtain soliton solutions. In Section~\ref{sec:pure} we restrict 
out consideration to the subclass of ``pure'' soliton solutions. This essentially disregards solutions with 
substructures of the form of Miles resonances. Section~\ref{sec:trop} addresses the tropical limit of pure soliton 
solutions. The cases of two and three solitons are then treated in Sections~\ref{sec:2solitons} and 
\ref{sec:3solitons}. Section~\ref{sec:vector_case_and_R-matrix} provides a general proof of the fact that, 
in the vector case, an $R$-matrix relates the polarizations at crossings. The linearity of the Yang-Baxter map 
in the vector case is certainly related to the particularly simple structure of the vector pKP equation  
mentioned above. 
In Section~\ref{sec:recon} we show how to construct a pure $N$-soliton solution of 
the vector KP equation from a pure $N$-soliton solution of the scalar KP equation, $N$ vector data and 
the aforementioned $R$-matrix. Section~\ref{sec:sol_tetrahedron} extends our exploration of the vector KP 3-soliton case 
and presents an apparently new solution of the tetrahedron (Zamolodchikov) equation (see, e.g., \cite{DMH15} 
and references cited there). In Section~\ref{sec:genR&tetra} we reveal the structure of the vector KP $R$-matrix,   
which leads us to a more general two-parameter $R$-matrix. Its parameter-dependence determines, via a ``local'' 
Yang-Baxter equation \cite{Mail+Nijh89PLB} (also see \cite{DMH15}), a solution of the functional tetrahedron equation 
(see, e.g., \cite{Serg98,KKS98,DMH15}), i.e., the set-theoretical version of the tetrahedron equation. 
Finally, Section~\ref{sec:conclusions} contains some concluding remarks.

\section{Soliton solutions of the $K$-modified matrix KP equation}  
\label{sec:sol}
The following describes a binary Darboux transformation for the pKP$_K$ equation (\ref{KmatrixpKP}). 
This is a simple extension of what is presented in \cite{CDMH16}, for example. 
Let $\phi_0$ be a solution of (\ref{KmatrixpKP}). 
Let $\theta$ and $\chi$ be $m \times N$, respectively $N \times n$, matrix solutions of the linear equations  
\bez
  && \theta_y = \theta_{xx} + 2 \phi_{0,x} K \theta \, , \qquad
     \theta_t = \theta_{xxx} + 3 \phi_{0,x} K \theta_x + \frac{3}{2} (\phi_{0,y} + \phi_{0,xx}) K \theta \, , \\
  && \chi_y = - \chi_{xx} - 2 \chi K \phi_{0,x} \, , \qquad
     \chi_t = \chi_{xxx} + 3 \chi_x K \phi_{0,x} - \frac{3}{2} \chi K (\phi_{0,y} - \phi_{0,xx}) \, .
\eez  
Then the system
\be
   \Omega_x = - \chi K \theta \, , \quad
   \Omega_y = - \chi K \theta_x + \chi_x K \theta \, , \quad
   \Omega_t = - \chi K \theta_{xx} + \chi_x K \theta_x - \chi_{xx} K \theta - 3 \chi K \phi_{0,x} K \theta \, ,
         \label{Omega_eqs}
\ee
is compatible and can thus be integrated to yield an $N \times N$ matrix solution $\Omega$. 
If $\Omega$ is invertible, then 
\be
    \phi = \phi_0 - \theta \, \Omega^{-1} \chi   \label{phi_new_solution}
\ee
is a new solution of (\ref{KmatrixpKP}). 

For vanishing\footnote{More generally, the following holds for any \emph{constant} $\phi_0$. But adding to
$\phi$ a constant matrix is an obvious symmetry of the pKP$_K$ equation.} 
seed solution, i.e., $\phi_0=0$, soliton solutions are obtained as follows. Let 
\bez
    \theta = \sum_{a=1}^A \theta_a \, e^{\vartheta(P_a)} \, , \qquad
  \chi = \sum_{b=1}^B e^{-\vartheta(Q_b)} \, \chi_b  \, ,  
\eez  
where $P_a, Q_b$ are constant $N \times N$ matrices, $\theta_a, \chi_b$ are constant $m \times N$, 
respectively $N \times n$ matrices, and
\be
     \vartheta(P) = x \, P + y \, P^2 + t \, P^3 \, .   \label{vartheta}
\ee
If, for all $a,b$, the matrices $P_a$ and $Q_b$ have no eigenvalue in common,  
there are unique $N \times N$ matrix solutions $W_{ba}$ of the Sylvester equations
\bez
     Q_b W_{ba} - W_{ba} P_a = \chi_b K \theta_a   \qquad a=1,\ldots,A, \quad b=1,\ldots,B \, .
\eez
Then (\ref{Omega_eqs}) is solved by 
\bez
    \Omega = \Omega_0 + \sum_{a=1}^A \sum_{b=1}^B e^{-\vartheta(Q_b)} \, W_{ba} \, e^{\vartheta(P_a)} \, ,
\eez
with a constant $N \times N$ matrix $\Omega_0$, 
and (\ref{phi_new_solution}) determines a soliton solution of (\ref{KmatrixpKP}) (and thus via $u = 2 \, \phi_x$ 
a solution of (\ref{KmatrixKP})), if $\Omega$ is everywhere invertible. 

\begin{remark}
\label{rem:hierarchy}
Corresponding solutions of the pKP$_K$ \emph{hierarchy} are obtained by replacing (\ref{vartheta}) with
$\vartheta(P) = \sum_{r=1}^\infty t_r \, P^r$, where $t_1=x$, $t_2=y$, $t_3=t$. 
\hfill $\square$
\end{remark}

\section{Pure soliton solutions}
\label{sec:pure}
In the following we restrict our considerations to the case where $A=B=1$. Then    
there remains only a single Sylvester equation,
\bez
     Q_1 W - W P_1 = \chi_1 K \theta_1 \, .
\eez
Moreover, we will restrict the matrices $P_1$ and $Q_1$ to be diagonal. It is convenient to name the diagonal entries
(``spectral parameters'') in two different ways, 
\bez
   &&  P_1 = \mathrm{diag}(p_{1,1},\ldots,p_{N,1}) = \mathrm{diag}(p_1,\ldots,p_N) \, , \\
   &&  Q_1 = \mathrm{diag}(p_{1,2},\ldots,p_{N,2}) = \mathrm{diag}(q_1,\ldots,q_N)  \, .
\eez
We further write
\bez
    \theta_1 = \left(\begin{array}{cccc} (q_1-p_1) \xi_1 & (q_2-p_2) \xi_2 
          & \cdots & (q_N-p_N) \xi_N \end{array} \right) \, , \qquad
    \chi_1 = \left(\begin{array}{c} \eta_1 \\ \vdots \\ \eta_N \end{array} \right) \, ,
\eez
where $\xi_i$ are $m$-component column vectors and $\eta_i$ are $n$-component row vectors.
Then the solution of the above Sylvester equation is given by 
\bez
     W = (w_{ij}) \, , \qquad
     w_{ij} = \frac{q_j - p_j}{q_i - p_j} \, \eta_i \, K \, \xi_j \qquad i,j=1,\ldots,N \, .
\eez
Furthermore, we set $\Omega_0 = I_N$, the $N \times N$ identity matrix. Hence $\Omega = (\Omega_{ij})$ with
\be
     \Omega_{ij} = \delta_{ij} + w_{ij} \, e^{\vartheta(p_j) - \vartheta(q_i)} \, ,  \label{Omega}
\ee
where $\delta_{ij}$ is the Kronecker delta. 
We call soliton solutions obtained from (\ref{phi_new_solution}), with the above restrictions, ``pure solitons''.
All what follows refers to them. 

We introduce 
\bez
    \vartheta_I := \sum_{i=1}^N \vartheta(p_{i,a_i})   \qquad \mbox{if} \quad
    I = (a_1,\ldots,a_N) \in \{1,2\}^N  \, .
\eez
Instead of using $(a_1,\ldots,a_N)$ as a subscript (or superscript), we will simply write $a_1 \ldots a_N$ 
in the following. For example, $\vartheta_{a_1 \ldots a_N} = \vartheta_{(a_1,\ldots,a_N)}$. 

 From (\ref{phi_new_solution}) we find that the pure soliton solutions of the pKP$_K$ equation are given by
\be
    \phi = \frac{F}{\tau} \, ,   \label{phi=F/tau}
\ee
with
\be
    \tau &:=& e^{\vartheta_{\boldsymbol{2}}} \, \det \Omega  \, ,   \label{tau_pure}  \\
    F &:=& - e^{\vartheta_{\boldsymbol{2}}} \, \theta_1 \, e^{\vartheta(P_1)} \, 
          \mathrm{adj}(\Omega) \, e^{-\vartheta(Q_1)} \, \chi_1 \, ,    \label{F_pure}
\ee          
where $\mathrm{adj}(\Omega)$ denotes the adjugate of the matrix $\Omega$ and $\boldsymbol{2} := 2\ldots 2 = (2,\ldots,2)$.   

\begin{proposition}
\label{prop:tau,F_expansion}
$\tau$ and $F$ have expansions
\be
      \tau = \sum_{I \in \{1,2\}^N} \mu_I \, e^{\vartheta_I} \, ,    \label{tau_pure_expansion}  \\
       F = \sum_{I \in \{1,2\}^N} M_I \, e^{\vartheta_I} \, , \label{F_pure_expansion}
\ee
with constants $\mu_I$ and constant $m \times n$ matrices $M_I$, where $\mu_{\boldsymbol{2}}=1$ and $M_{\boldsymbol{2}}=0$. 
\end{proposition}
\begin{proof}
 From the definition of the determinant, 
$\det \Omega = \epsilon^{i_1 \ldots i_N} \, \Omega_{1 i_1} \cdots \Omega_{N i_N}$, with the Levi-Civita 
symbol $\epsilon^{i_1 \ldots i_N}$ and summation convention, we know that $\det \Omega$  
consists of a sum of monomials of order $N$ in the entries $\Omega_{ij}$. Here the latter are given by (\ref{Omega}). 
If no diagonal term $\Omega_{ii}$ appears in a monomial, its phase factor is 
$e^{\vartheta_{\boldsymbol{1}} - \vartheta_{\boldsymbol{2}}}$. If one diagonal entry 
$\Omega_{ii} = 1 + w_{ii} \, e^{\vartheta(p_{i,1}) - \vartheta(p_{i,2})}$ appears in a monomial,  
the latter splits into two parts. Only the part arising from the summand $1$ is different as now  
the phase factor is $e^{\vartheta_{\boldsymbol{1}} - \vartheta_{\boldsymbol{2}} - \vartheta(p_{i,1}) + \vartheta(p_{i,2})}$. 
From monomials containing several diagonal entries of $\Omega$, we obtain summands with a phase factor 
of the form 
\bez
  e^{\vartheta_{\boldsymbol{1}} - \vartheta_{\boldsymbol{2}} - \vartheta(p_{i_1,1}) +\vartheta(p_{i_1,2})
 + \cdots - \vartheta(p_{i_r,1}) + \vartheta(p_{i_r,2})}
 = e^{[\vartheta_{\boldsymbol{1}} - \vartheta(p_{i_1,1}) - \cdots - \vartheta(p_{i_r,1}) ]
   + \vartheta(p_{i_1,2}) + \cdots + \vartheta(p_{i_r,2}) - \vartheta_{\boldsymbol{2}} } \, .
\eez
Finally, from a monomial with $N$ diagonal entries of $\Omega$, we also obtain a constant term, namely $1$. 
Now our assertion (\ref{tau_pure_expansion}) follows since $\tau$ is $\det \Omega$ multiplied by $e^{\vartheta_{\boldsymbol{2}}}$. 
Clearly, $\mu_{\boldsymbol{2}}=1$. 

According to the Laplace (cofactor) expansion $\det \Omega = \sum_{j=1}^N \Omega_{ij} \, \mathrm{adj}(\Omega)_{ji}$ 
with respect to the $i$-th row, the term
$\Omega_{ij} \, \mathrm{adj}(\Omega)_{ji}$ consists of all summands in $\det \Omega$ having $\Omega_{ij}$ as a factor. 
(\ref{tau_pure_expansion}) implies that a summand of $e^{\vartheta_{\boldsymbol{2}}} \mathrm{adj}(\Omega)_{ji}$ then has 
a phase factor of the form $e^{\vartheta_I - \vartheta(p_{j,1}) + \vartheta(p_{i,2})}$, with some $I \in \{1,2\}^N$, so that 
$e^{\vartheta_{\boldsymbol{2}}} \, (e^{\vartheta(P_1)} \, \mathrm{adj}(\Omega) \, e^{-\vartheta(Q_1)})_{ji}$ has 
the phase factor $e^{\vartheta_I}$. Hence (\ref{F_pure_expansion}) holds.
Furthermore, no entry of $e^{\vartheta(P_1)} \, \mathrm{adj}(\Omega) \, e^{-\vartheta(Q_1)}$ is constant, hence $M_{\boldsymbol{2}}=0$. 
\end{proof}

\begin{remark}
The introduction of the redundant factor $e^{\vartheta_{\boldsymbol{2}}}$ in (\ref{phi=F/tau}), via the 
definitions (\ref{tau_pure}) and (\ref{F_pure}), achieves that $\tau$ and $F$ are linear combinations 
of exponentials $e^{\vartheta_I}$, $I \in \{1,2\}^N$, in which case we have a very convenient labelling. 
This is also so if we choose the factor $e^{-\vartheta_{\boldsymbol{1}}}$ instead, which  
leads to an expansion in terms of $e^{-\vartheta_I}$, now with $M_{\boldsymbol{1}}=0$. 
\hfill $\square$
\end{remark}

Regularity of a pure soliton solution requires $\mu_I \geq 0$ for all $I \in \{1,2\}^N$ (or 
equivalently $\mu_I \leq 0$ for all $I \in \{1,2\}^N$) and $\mu_I \neq 0$ for at least one $I$. 
If $\mu_I =0$ for some $I$, this means that the phase 
$\vartheta_I$ is not present in the expression for $\tau$. In this case one has to arrange the data 
in such a way that $M_I =0$ in order to avoid unbounded exponential growth of the soliton solution 
in some phase region. But we will disregard such cases and add the condition $\mu_I > 0$, $\forall I \in \{1,2\}^N$, 
to our definition of pure soliton solutions. 

It follows that the corresponding solution of the KP equation is given by  
\bez
     u = \frac{1}{\tau^2} \sum_{I,J \in \{1,2\}^N}  
         (p_J - p_I) (\mu_I M_J - \mu_J M_I) \, e^{\vartheta_I} \, e^{\vartheta_J} \, ,
\eez
where 
\bez
      p_I = p_{1,a_1} + \cdots + p_{N,a_N} \quad \mbox{if} \quad
      I = (a_1,\ldots,a_N) \, . 
\eez   
Using Jacobi's formula for the derivative of a determinant, we obtain
\bez
    \tau_x 
 &=& p_{\boldsymbol{2}} \, \tau + e^{\vartheta_{\boldsymbol{2}}} 
     \, \mathrm{tr}(\mathrm{adj}(\Omega) \, \Omega_x) 
 = p_{\boldsymbol{2}} \, \tau - e^{\vartheta_{\boldsymbol{2}}} \, \mathrm{tr}(\mathrm{adj}(\Omega) \, \chi K \theta) 
 = p_{\boldsymbol{2}} \, \tau - e^{\vartheta_{\boldsymbol{2}}} \, \mathrm{tr}(K \theta \, \mathrm{adj}(\Omega) \, \chi) 
   \nonumber \\ 
 &=& p_{\boldsymbol{2}} \, \tau + \mathrm{tr}(K \, F) \, ,
\eez
which implies
\be
     \mathrm{tr}(K \phi) = (\ln \tau)_x - p_{\boldsymbol{2}} \, , \label{tr(Kphi)}
\ee
and thus
\bez
     \mathrm{tr}(K u) = 2 \, (\ln \tau)_{xx} \, .
\eez
Using (\ref{phi=F/tau}) in (\ref{tr(Kphi)}), and reading off the coefficient of $e^{\vartheta_I}$, we find
\be
      \mathrm{tr}(K \, M_I) = (p_I - p_{\boldsymbol{2}}) \, \mu_I \, .   \label{trKM}
\ee

\begin{remark}
If $n=1$, (\ref{tr(Kphi)}) reads 
\bez
     K \phi = (\ln \tau)_x - p_{\boldsymbol{2}} \, ,
\eez
and  
\bez 
        K u = 2 \, (\ln \tau)_{xx} 
\eez        
is a solution of the scalar KP equation. If $n>1$, $\mathrm{tr}(K u)$ is not in general a solution 
of the scalar KP equation. 
\hfill $\square$
\end{remark}

\begin{remark}
\label{rem:KdV_reduction} 
Dropping the redundant factor $e^{\vartheta_{\boldsymbol{2}}}$ in (\ref{tau_pure}) and (\ref{F_pure}) means 
that we have to multiply the above expressions (\ref{tau_pure_expansion}) and (\ref{F_pure_expansion}) for 
$\tau$ and $F$ by $e^{-\vartheta_{\boldsymbol{2}}}$. 
It is then evident that $\phi$ only depends on differences of phases of the form 
$\vartheta(p_{i,1}) - \vartheta(p_{i,2}) = (p_{i,1} - p_{i,2}) \, x + (p_{i,1}^2 - p_{i,2}^2) \, y + (p_{i,1}^3 - p_{i,2}^3) \, t$. 
As a consequence, setting $p_{i,2} = - p_{i,1}$, i.e., $q_i = -p_i$, eliminates the $y$-terms in all phases. 
This means that, under this condition for the parameters, we could have started as well with 
$\vartheta(P) = x \, P + t \, P^3$, hence without the $y$-term in (\ref{vartheta}). In this way we make 
contact with the KdV$_K$ reduction of KP$_K$.  
\hfill $\square$
\end{remark}

\section{Tropical limit of pure soliton solutions}
\label{sec:trop}
A crucial point is that we define the tropical limit of the matrix soliton solution via the tropical 
limit of the scalar function $\tau$ (cf. \cite{DMH11KPT,DMH12KPBT,DMH14KdV}). Let
\be
    \phi_I := \phi\Big|_{\vartheta_J \to -\infty, J \neq I} = \frac{M_I}{\mu_I} \, . \label{phi_I}
\ee
In a region where a phase $\vartheta_I$ dominates all others, in the sense that 
$\log(\mu_I \, e^{\vartheta_I}) > \log(\mu_J \, e^{\vartheta_J})$ for all participating $J \neq I$, 
the tropical limit of the potential $\phi$ is given by (\ref{phi_I}). 
It should be noticed that these expressions do not depend on the coordinates $x,y,t$. 

The boundary between the regions associated with the phases $\vartheta_I$ and $\vartheta_J$ 
is determined by the condition 
\be
     \mu_I \, e^{\vartheta_I} = \mu_J \, e^{\vartheta_J} \, .    \label{phase_boundary}
\ee
Not all parts of such a boundary are visible at fixed time, since some of them may lie in a region where 
a third phase dominates the two phases. The tropical limit of a soliton solution at a fixed time $t$ has support on the 
visible parts of the boundaries between the regions associated with phases appearing in $\tau$. 
On such a visible boundary segment, the value of $u$ is given by
\bez
     u_{IJ} = \frac{1}{2} (p_I - p_J) \left( \phi_I - \phi_J \right)  \, .
\eez

For $I = (a_1,\ldots,a_N)$ we set
\bez
     I_k(a) = (a_1,\ldots,a_{k-1},a,a_{k+1}, \ldots,a_N) \, .
\eez
At fixed time, the set of line segments associated with the $k$-th soliton are obtained from (\ref{phase_boundary}) with 
$I = I_k(1)$ and $J=I_k(2)$, for all possible $I$. They satisfy
\be
     (p_{k,2} - p_{k,1}) \, x + (p_{k,2}^2 - p_{k,1}^2) \, y + (p_{k,2}^3 - p_{k,1}^3) \, t 
     + \ln \frac{\mu_{I_k(2)}}{\mu_{I_k(1)}} 
     = 0 \, .    \label{kth_soliton_lines}
\ee
All these line segments have the same slope $-(p_{k,2}+p_{k,1})^{-1}$ in the $xy$-plane, hence they are parallel. 
The shifts between them are given by 
\bez
    \delta^{(k)}_{IJ} = \ln \Big( 
    \frac{\mu_{I_k(2)}}{\mu_{I_k(1)}}
    \frac{\mu_{J_k(1)}}{\mu_{J_k(2)}}
    \Big) \, .
\eez
They give rise to the familiar asymptotic ``phase shifts'' of line solitons. 
The tropical limit of $u$ on a visible line segment of the $k$-th soliton is given by
\bez
    u_{I_k(1) \, I_k(2)}
  = \frac{1}{2} (p_{k,1} - p_{k,2}) \left( \phi_{I_k(1)} - \phi_{I_k(2)} \right) \, .
\eez
The value of $u$ at a visible triple phase coincidence is
\bez
  u_{IJL} = \frac{4}{9} ( u_{IJ} + u_{IL} + u_{JL} ) \
            = \frac{2}{9} \Big( (2 p_I - p_J - p_L) \phi_I + (2 p_J - p_I - p_L) \phi_J + (2 p_L - p_I - p_J) \phi_L \Big) \, .
\eez
Instead of the above expressions for the tropical values of $u$ we will rather consider
\be
    \hat{u}_{IJ} = \frac{\phi_I - \phi_J}{p_I - p_J}  \, ,  \label{hatu_IJ}
\ee
which has the form of a discrete derivative. 
Since (\ref{trKM}) and (\ref{phi_I}) imply 
\be
     \mathrm{tr}(K \phi_I) = p_I - p_{\boldsymbol{2}} \, ,  \label{tr(K phi_I)}
\ee
the latter values are normalized in the sense that
\be
     \mathrm{tr}(K \hat{u}_{IJ}) = 1 \, .    \label{trKu}
\ee

If $I=(a_1,\ldots,a_N)$ and $i \neq j$, let
\bez
     I_{ij}(a,b) = (a_1,\ldots,a_N)\Big|_{a_i \mapsto a,a_j \mapsto b} \, .
\eez
The normalized tropical values of $u$ satisfy 
\bez
  &&\hspace{-.8cm}  ( p_i - q_i ) \, \hat{u}_{I_{ij}(1,1) \, I_{ij}(2,1)}  
     + ( p_j - q_j ) \, \hat{u}_{I_{ij}(2,1) \, I_{ij}(2,2)}
   = (p_i - q_i + p_j - q_j) \, \hat{u}_{I_{ij}(1,1) \, I_{ij}(2,2)} \, , \\  
  &&\hspace{-.8cm} ( p_i - q_i ) \, \hat{u}_{I_{ij}(1,2) \, I_{ij}(2,2)}  
     + (p_j - q_j ) \, \hat{u}_{I_{ij}(1,1) \, I_{ij}(1,2)}  
  = ( p_i - q_i + p_j - q_j ) \, \hat{u}_{I_{ij}(1,1) \, I_{ij}(2,2)} \, , \\
  &&\hspace{-.8cm} ( p_i - q_i ) \, \hat{u}_{I_{ij}(1,1) \, I_{ij}(2,1)}  
     + (p_j - q_j ) \, \hat{u}_{I_{ij}(1,1),I_{ij}(1,2)}  
  = ( p_i - q_i + q_j - p_j ) \, \hat{u}_{I_{ij}(1,2) \, I_{ij}(2,1)} \, , \\
  &&\hspace{-.8cm} ( p_i - q_i ) \, \hat{u}_{I_{ij}(1,2) \, I_{ij}(2,2)}  
     + (p_j - q_j ) \, \hat{u}_{I_{ij}(2,1),I_{ij}(2,2)}  
  = ( p_i - q_i + q_j - p_j ) \, \hat{u}_{I_{ij}(1,2) \, I_{ij}(2,1)} \, .
\eez
These identities are simply consequences of the definition of $\hat{u}_{IJ}$. They linearly relate the 
(normalized) polarizations at points of the tropical limit graph, where three lines meet.

\section{Pure 2-soliton solutions}
\label{sec:2solitons}
Let $N=2$. Then we have 
\bez
    \tau = e^{\vartheta_{22}} + \kappa_{11} \, e^{\vartheta_{12}} + \kappa_{22} \, e^{\vartheta_{21}} 
           + \alpha \, \kappa_{11} \kappa_{22} \, e^{\vartheta_{11}} \, ,           
\eez
where
\bez
     \kappa_{ij} = \eta_i K \xi_j \, , \qquad 
    \alpha = 1 - \frac{(q_1-p_1)(q_2-p_2) \, \kappa_{12} \, \kappa_{21}}{(q_1-p_2)(q_2-p_1) \, \kappa_{11} \, \kappa_{22}} \, ,
\eez
and 
\bez
    F &=& (q_1-p_1)(q_2-p_2) \, \Big( \frac{\kappa_{22}}{p_2-q_2} \, \xi_1 \otimes \eta_1 
        + \frac{\kappa_{11}}{p_1-q_1} \, \xi_2 \otimes \eta_2 
        + \frac{\kappa_{12}}{q_1-p_2} \, \xi_1 \otimes \eta_2 \\
      && + \frac{\kappa_{21}}{q_2-p_1} \, \xi_2 \otimes \eta_1 \Big) e^{\vartheta_{11}} 
        + (p_1-q_1) \, \xi_1 \otimes \eta_1 \, e^{\vartheta_{12}}
        + (p_2-q_2) \, \xi_2 \otimes \eta_2 \, e^{\vartheta_{21}} \, .
\eez
The tropical values of the pKP$_K$ solution $\phi$ in the dominant phase regions are then given by 
\bez
   &&  \phi_{11} = \frac{(q_1-p_1)(q_2-p_2)}{\alpha \, \kappa_{11} \, \kappa_{22}} 
                  \Big( \frac{\kappa_{22}}{ p_2-q_2} \, \xi_1 \otimes \eta_1
                   + \frac{\kappa_{11}}{ p_1-q_1} \, \xi_2 \otimes \eta_2 \\
   && \hspace{1.5cm}   + \frac{\kappa_{12}}{ q_1-p_2 } \, \xi_1 \otimes \eta_2
                   + \frac{\kappa_{21}}{ q_2-p_1 } \, \xi_2 \otimes \eta_1 \Big) \, , \\
 && \phi_{12} = (p_1-q_1) \frac{\xi_1 \otimes \eta_1}{\kappa_{11}} \, , \quad
    \phi_{21} = (p_2-q_2) \frac{\xi_2 \otimes \eta_2}{\kappa_{22}} \, , \quad
    \phi_{22} = 0 \, .
\eez

\begin{remark}
The above values $\phi_{ab}$ solve the following nonlinear equation,
\be
 && \Big(1 + \mathrm{tr} \frac{K (\phi_{12} - \phi_{22}) K (\phi_{21}-\phi_{22})}{(q_1-p_2) (p_1-q_2)} \Big) 
           (\phi_{11} - \phi_{22})  
  - \frac{(\phi_{12} - \phi_{22}) K (\phi_{21} - \phi_{22})}{q_1-p_2} \nonumber \\
 && + \frac{(\phi_{21} - \phi_{22}) K (\phi_{12} - \phi_{22})}{p_1-q_2} - (\phi_{12} - \phi_{22}) 
   - (\phi_{21} - \phi_{22}) = 0 \, .   \label{nonlin_phi_eq}
\ee
Addressing more than two solitons, non-zero counterparts of $\phi_{22}$ will show up, as displayed  
in this equation.
\hfill $\square$
\end{remark}

For the tropical values of $\hat{u}$ along the phase region boundaries, we obtain
\be
 && u_{1,\mathrm{in}} := \hat{u}_{11,21} 
    = \alpha^{-1} \Big(1_m - \frac{q_2-p_2}{q_2-p_1} \, \frac{\xi_2 \otimes \eta_2}{\kappa_{22}} K \Big)
     \frac{\xi_1 \otimes \eta_1}{\kappa_{11}} \Big( 1_n - \frac{q_2-p_2}{q_1-p_2} K \frac{\xi_2 \otimes \eta_2}{\kappa_{22}} \Big)
      \, , \nonumber \\
 && u_{2,\mathrm{in}} := \hat{u}_{21,22} = \frac{\xi_2 \otimes \eta_2}{\kappa_{22}} \, , \nonumber \\
 && u_{1,\mathrm{out}} := \hat{u}_{12,22}  = \frac{\xi_1 \otimes \eta_1}{\kappa_{11}} \, , \nonumber \\
 && u_{2,\mathrm{out}} := \hat{u}_{11,12}  
    = \alpha^{-1} \Big(1_m - \frac{q_1-p_1}{q_1-p_2} \, \frac{\xi_1 \otimes \eta_1}{\kappa_{11}} K \Big)
     \frac{\xi_2 \otimes \eta_2}{\kappa_{22}} 
     \Big( 1_n - \frac{q_1-p_1}{q_2-p_1} K \frac{\xi_1 \otimes \eta_1}{\kappa_{11}} \Big) \, , \label{2s_u_i,in/out}
\ee
where $1_m$ stands for the $m \times m$ identity matrix. For the in/out classification, see Fig.~\ref{fig:3x2} below.  
All the matrices in (\ref{2s_u_i,in/out}) have rank one, which is not at all obvious from the form of $\phi_{ab}$. 
We obtain the following nonlinear relation between ``incoming'' and ``outgoing'' polarizations,
\be
  &&  u_{1,\mathrm{out}} = \alpha_{\mathrm{in}}^{-1} \Big( 1_m - \frac{q_2-p_2}{p_1-p_2} u_{2,\mathrm{in}} K \Big) u_{1,\mathrm{in}} 
           \Big(1_n - \frac{p_2-q_2}{q_1-q_2} K u_{2,\mathrm{in}} \Big)  \, , \nonumber \\
  &&  u_{2,\mathrm{out}} = \alpha_{\mathrm{in}}^{-1} \Big( 1_m - \frac{q_1-p_1}{q_1-q_2} u_{1,\mathrm{in}} K \Big) u_{2,\mathrm{in}} 
         \Big( 1_n - \frac{p_1-q_1}{p_1-p_2} K u_{1,\mathrm{in}} \Big) \, , 
         \label{YBmap_u_in->u_out}
\ee
where 
\bez
  \alpha _{\mathrm{in}} = 1 
    - \frac{(p_1-q_1)(p_2-q_2)}{(p_1-p_2)(q_1-q_2)} \mathrm{tr}\left( K u_{1,\mathrm{in}} K u_{2,\mathrm{in}} \right) \, .
\eez
We note that $\alpha \, \alpha _{\mathrm{in}} = 1$. 
(\ref{YBmap_u_in->u_out}) is a new nonlinear Yang-Baxter map with parameters $p_i,q_i$, $i=1,2$. 

Writing
\bez
    u_{a,\mathrm{in}} = \frac{\xi_{a,\mathrm{in}} \otimes \eta_{a,\mathrm{in}}}{\eta_{a,\mathrm{in}} K \xi_{a,\mathrm{in}}} 
      \, , \quad
    u_{a,\mathrm{out}} = \frac{\xi_{a,\mathrm{out}} \otimes \eta_{a,\mathrm{out}}}{\eta_{a,\mathrm{out}} K \xi_{a,\mathrm{out}}} 
    \qquad a=1,2 \, ,
\eez
determines $\xi_{1,\mathrm{in/out}}$ and $\eta_{1,\mathrm{in/out}}$ up to scalings. 
We find
\bez
  && \xi_{1,\mathrm{out}} = \alpha_{\mathrm{in}}^{-1/2} \Big( 1_m - \frac{p_2 - q_2}{p_2-p_1} \, 
   \frac{\xi_{2,\mathrm{in}} \otimes \eta_{2,\mathrm{in}}}{\eta_{2,\mathrm{in}} K \xi_{2,\mathrm{in}}} K \Big)
   \xi_{1,\mathrm{in}} \, , \\
  && \xi_{2,\mathrm{out}} = \alpha_{\mathrm{in}}^{-1/2} \Big( 1_m - \frac{q_1 - p_1}{q_1-q_2} \, 
   \frac{\xi_{1,\mathrm{in}} \otimes \eta_{1,\mathrm{in}}}{\eta_{1,\mathrm{in}} K \xi_{1,\mathrm{in}}} K \Big)
   \xi_{2,\mathrm{in}} \, , \\
  && \eta_{1,\mathrm{out}} = \alpha_{\mathrm{in}}^{-1/2} \eta_{1,\mathrm{in}} 
    \Big( 1_n - \frac{q_2 - p_2}{q_2-q_1} \, K \, \frac{\xi_{2,\mathrm{in}} \otimes \eta_{2,\mathrm{in}}}{\eta_{2,\mathrm{in}} K \xi_{2,\mathrm{in}}} \Big) 
    \, , \\
  && \eta_{2,\mathrm{out}} = \alpha_{\mathrm{in}}^{-1/2} \eta_{2,\mathrm{in}} 
    \Big( 1_n - \frac{p_1 - q_1}{p_1-p_2} \, K \, \frac{\xi_{1,\mathrm{in}} \otimes \eta_{1,\mathrm{in}}}{\eta_{1,\mathrm{in}} K \xi_{1,\mathrm{in}}} \Big) 
    \, ,
\eez
and
\bez
   ( \xi_{1,\mathrm{in}} , \eta_{1,\mathrm{in}} ; \xi_{2,\mathrm{in}} , \eta_{2,\mathrm{in}} ) 
      \mapsto  
   ( \xi_{1,\mathrm{out}} , \eta_{1,\mathrm{out}} ; \xi_{2,\mathrm{out}} , \eta_{2,\mathrm{out}}) 
\eez
is another form of the above Yang-Baxter map. 

\begin{remark}
In (\ref{2s_u_i,in/out}) we found that $u_{2,in}$ and $u_{1,out}$ have a simple elementary form. They are the polarizations 
at the two boundary lines of the dominating phase region numbered by $22 = (2,2)$, see Fig.~\ref{fig:3x2}.  
We know from Proposition~\ref{prop:tau,F_expansion} that it is special since $M_{22}=0$. Considering an ``evolution'' 
in negative $x$-direction (instead of $y$-direction), thus offers a more direct derivation of the Yang-Baxter map.  
\hfill $\square$
\end{remark}

\begin{example}
\label{ex:3x2}
Let $m=3$ and $n=2$. Choosing
\bez
     p_1 = - 3/4 \, , \quad
     p_2 = 1/4 \, , \quad
     q_1 = -1/4 \, , \quad
     q_2 = 3/4 \, ,
\eez
and
\bez
   \eta_1 = (1,0) \, , \quad 
   \eta_2 = (0,1) \, , \quad
   \xi_1 = \left( \begin{array}{c} 2/3 \\ 5 \\ -2 \end{array} \right) \, , \quad
   \xi_2 = \left( \begin{array}{c} 1 \\ 2/3 \\ 2 \end{array} \right) \, , \quad
   K = \left( \begin{array}{ccc} 1 & 1 & 1 \\ 1 & 2 & 1 \end{array} \right) \, ,
\eez
we obtain the first contour plot, at $t=0$, shown in Fig.~\ref{fig:3x2}. Fig.~\ref{fig:3x2_u-plots} shows 
plots of the components of the transpose of $u$. 
Choosing instead $p_2$ close to $q_1$ reveals an ``inner structure'' of crossing points, see the second 
plot in Fig.~\ref{fig:3x2}. This is the (phase) shift mentioned in Section~\ref{sec:trop}.
\begin{figure} 
\begin{center}
\includegraphics[scale=.2]{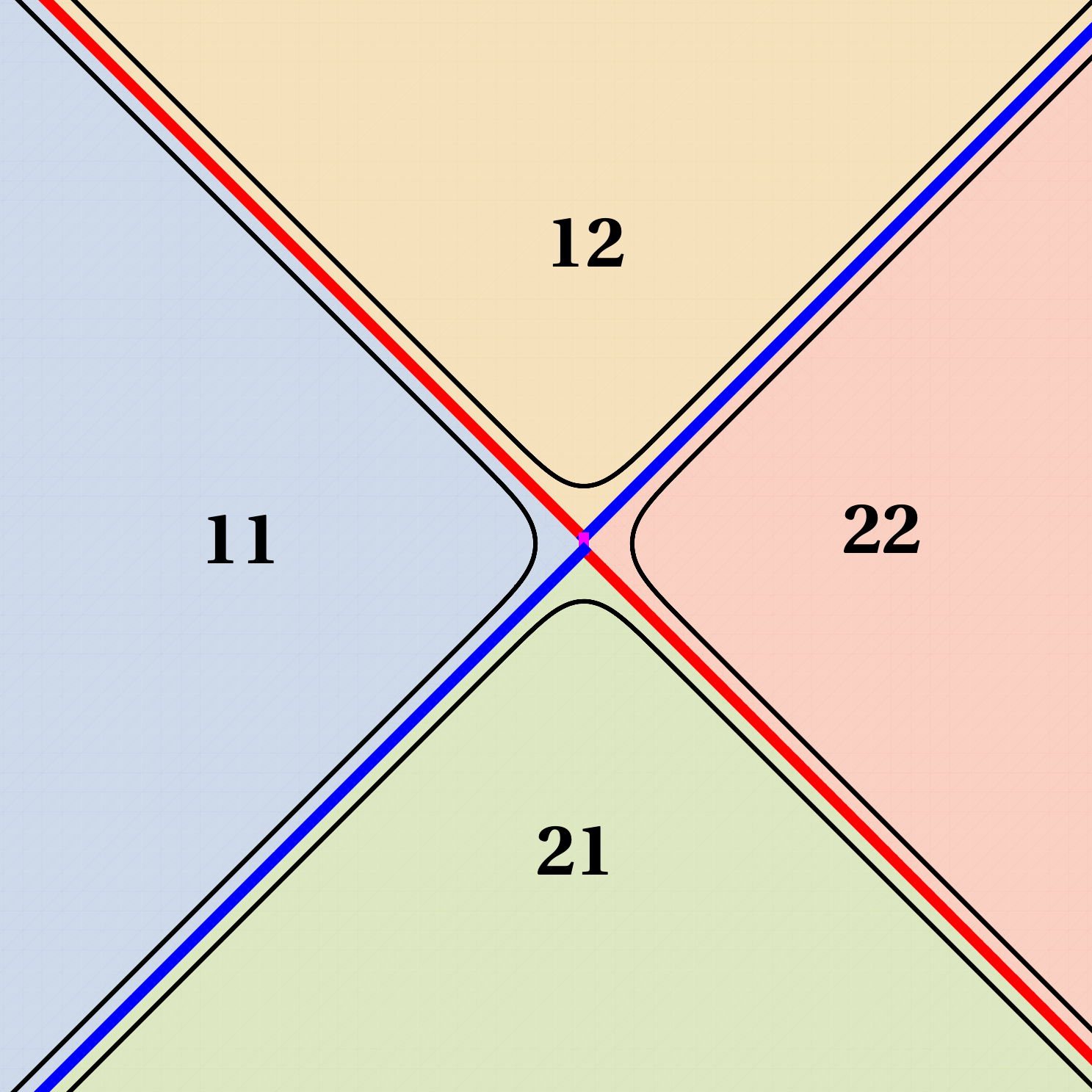} 
\hspace{1cm}
\includegraphics[scale=.2]{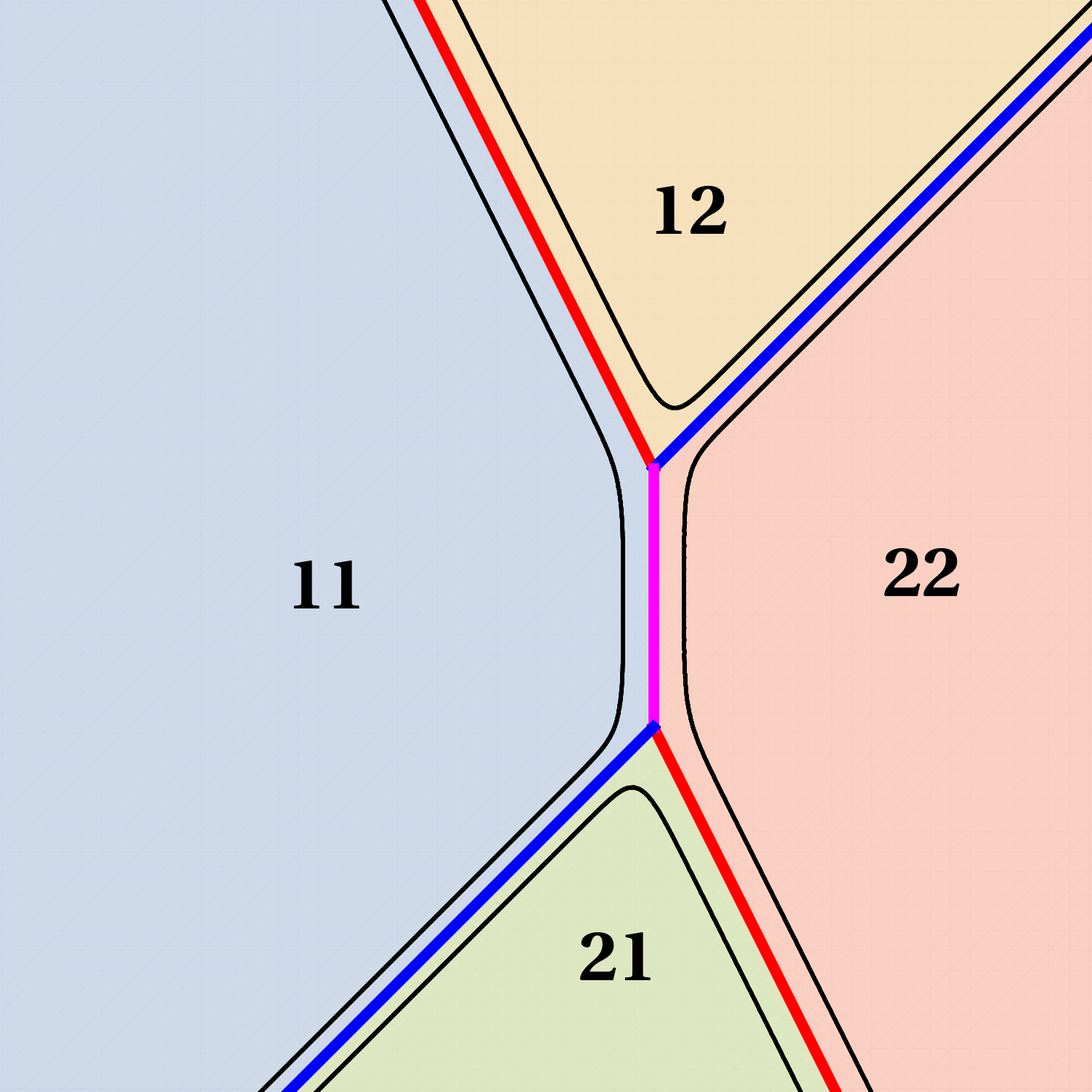} 
\hspace{1cm}
\includegraphics[scale=.2]{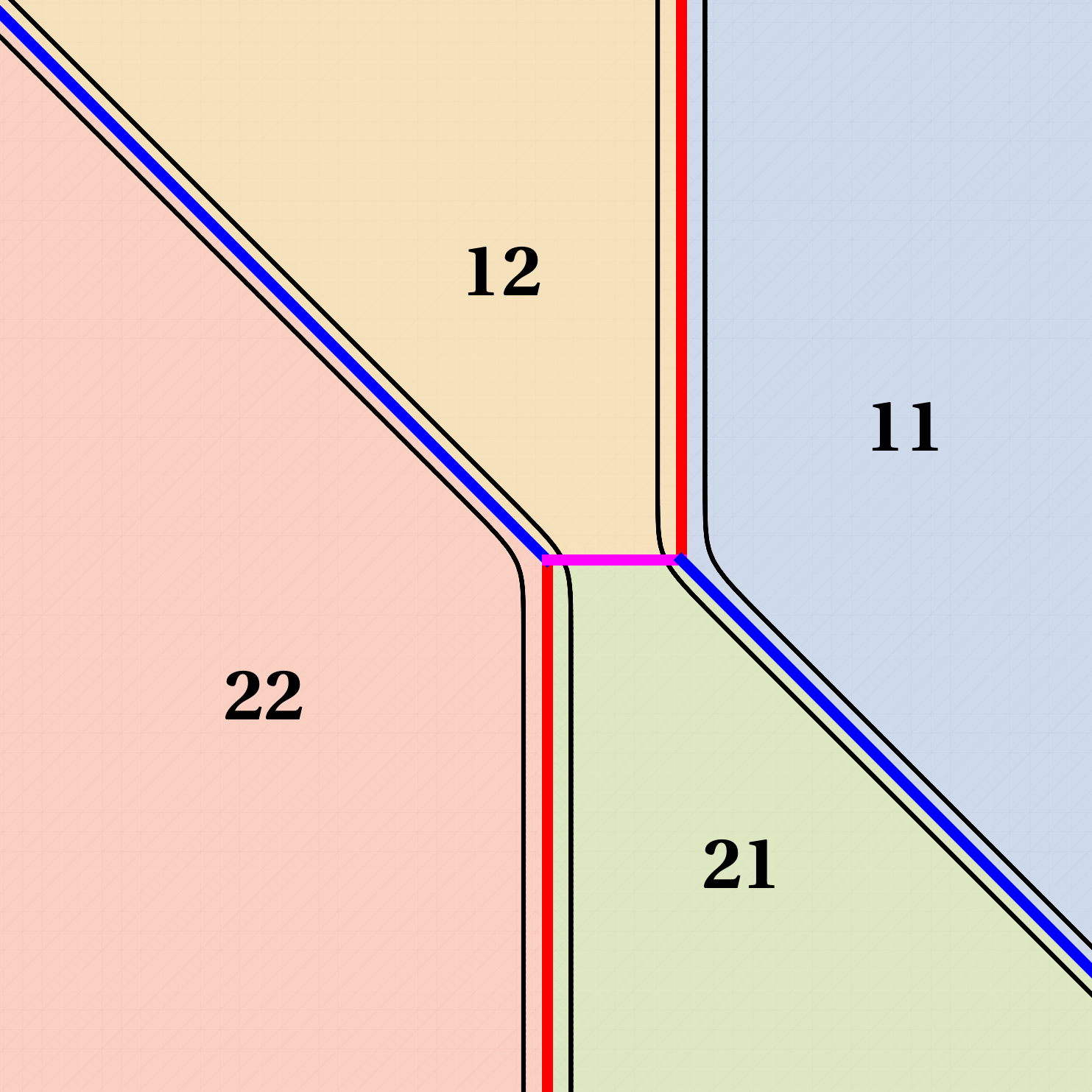} 
\end{center}
\caption{The first is a contour plot of $\mathrm{tr}(K u) = 2 (\ln \tau)_{xx}$ for a 2-soliton solution of 
the $2 \times 3$ matrix KP equation, at $t=0$ in the $xy$-plane, using the data of Example~\ref{ex:3x2}. 
Viewed as a process in $y$-direction, the YB map takes the values of 
the KP variable on the lower two legs to those on the upper two. A number $ij$ indicates the respective 
dominating phase region. 
In the second plot the value of $p_2$ is replaced by $-1/4+10^{-5}$, so that $p_2$ is very close to $q_1$. Here 
a boundary segment between phase regions $11$ and $22$ is visible.  
The third plot presents an example, where the parameters of the 2-soliton solution are now chosen such that 
the latter boundary is hidden and instead a boundary segment between phase regions $12$ and $21$ is visible.  
\label{fig:3x2} }
\end{figure}
\begin{figure} 
\begin{center}
\includegraphics[scale=.45]{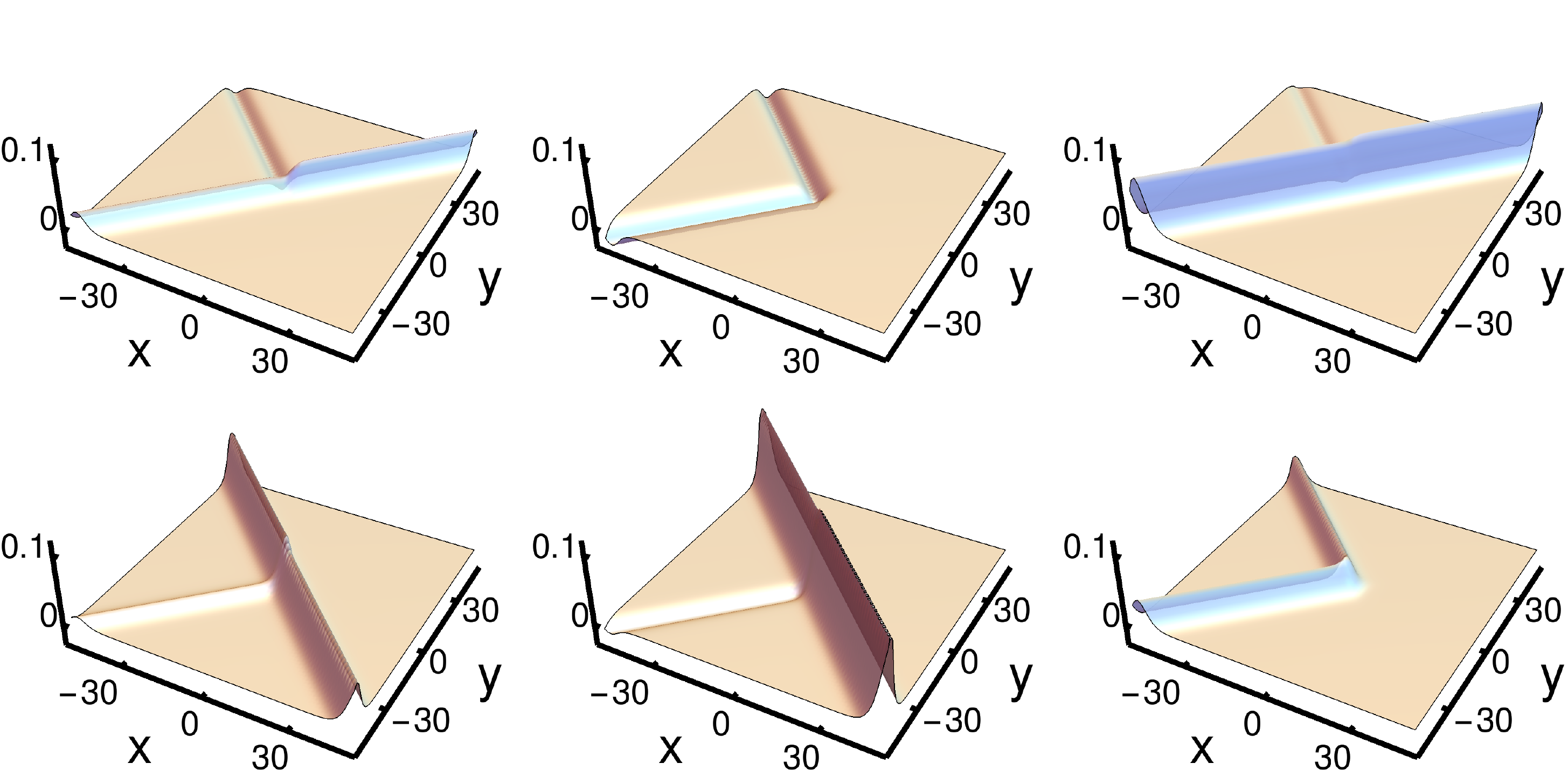} 
\end{center}
\caption{Plot of the six components of the transpose of $u$ at $t=0$ for the solution of the $3 \times 2$ 
matrix KP equation with the (first) data specified in Example~\ref{ex:3x2}. The components 
are localized exactly where $\mathrm{tr}(K u)$ is localized, cf. Fig.~\ref{fig:3x2}. 
\label{fig:3x2_u-plots} }
\end{figure}
\hfill $\square$
\end{example}

\begin{remark}
We should stress that the relevant structures are actually \emph{three}-dimensional and our figures only 
display a two-dimensional cross section. Instead of displaying structures in the $xy$-plane at constant $t$, 
we may as well look at those in the $xt$-plane at constant $y$. The latter becomes relevant if we consider 
the KdV$_K$ reduction.  
\hfill $\square$
\end{remark}

\begin{remark}
Soliton solutions of KdV$_K$ are obtained from those of KP$_K$ by setting $q_i = - p_i$, $i=1,\ldots,N$, 
see Remark~\ref{rem:KdV_reduction}.
Then the above equations reduce to
\bez
  && \xi_{1,\mathrm{out}} = \alpha_{\mathrm{in}}^{-1/2} \Big( 1_m - \frac{ 2 p_2}{p_2-p_1} \, 
   \frac{\xi_{2,\mathrm{in}} \otimes \eta_{2,\mathrm{in}}}{\eta_{2,\mathrm{in}} K \xi_{2,\mathrm{in}}} K \Big)
   \xi_{1,\mathrm{in}} \, , \\
  && \xi_{2,\mathrm{out}} = \alpha_{\mathrm{in}}^{-1/2} \Big( 1_m - \frac{2 p_1}{p_1-p_2} \, 
   \frac{\xi_{1,\mathrm{in}} \otimes \eta_{1,\mathrm{in}}}{\eta_{1,\mathrm{in}} K \xi_{1,\mathrm{in}}} K \Big)
   \xi_{2,\mathrm{in}} \, , \\
  && \eta_{1,\mathrm{out}} = \alpha_{\mathrm{in}}^{-1/2} \eta_{1,\mathrm{in}} 
    \Big( 1_n - \frac{2 p_2}{p_2-p_1} \, K \, \frac{\xi_{2,\mathrm{in}} \otimes \eta_{2,\mathrm{in}}}{\eta_{2,\mathrm{in}} K \xi_{2,\mathrm{in}}} \Big) 
    \, , \\
  && \eta_{2,\mathrm{out}} = \alpha_{\mathrm{in}}^{-1/2} \eta_{2,\mathrm{in}} 
    \Big( 1_n - \frac{2 p_1}{p_1-p_2} \, K \, \frac{\xi_{1,\mathrm{in}} \otimes \eta_{1,\mathrm{in}}}{\eta_{1,\mathrm{in}} K \xi_{1,\mathrm{in}}} \Big) 
    \, ,  
\eez
with 
\bez
    \alpha_{\mathrm{in}} = 1 + \frac{4 p_1 p_2}{(p_1-p_2)^2} \, 
    \frac{\eta_{1,\mathrm{in}} K \xi_{2,\mathrm{in}} \, \eta_{2,\mathrm{in}} K \xi_{1,\mathrm{in}}}{\eta_{1,\mathrm{in}} K 
      \xi_{1,\mathrm{in}} \, \eta_{2,\mathrm{in}} K \xi_{2,\mathrm{in}}} \, . 
\eez
If $K$ is the $N \times N$ identity matrix, this becomes the Yang-Baxter map first found by Veselov \cite{Vese03},  
also see \cite{Gonc+Vese04,Suri+Vese03}. 
The factor $\alpha_{\mathrm{in}}^{-1/2}$ is missing in these publications, but such a factor is necessary for the map 
to satisfy the Yang-Baxter equation. One can avoid the square root at the price of having an asymmetric 
appearance of factors $\alpha_{\mathrm{in}}^{-1}$. 
\hfill $\square$
\end{remark}

\subsection{Pure column vector 2-soliton solutions}
We set $n=1$. Now the $\eta_i$ are scalars and drop out of the relevant formulas. Introducing
\bez
     \hat{\xi}_i = \frac{\xi_i}{K \xi_i} \, ,
\eez
we have
\bez
    u_{1,\mathrm{in}} = \frac{p_1-q_2}{p_1-p_2} \hat{\xi}_1 + \frac{q_2-p_2}{p_1-p_2} \hat{\xi}_2 \, , \quad
    u_{2,\mathrm{in}} = \hat{\xi}_2 \, , \quad
    u_{1,\mathrm{out}} = \hat{\xi}_1 \, , \quad
    u_{2,\mathrm{out}} = \frac{p_1-q_1}{p_1-p_2} \hat{\xi}_1 + \frac{q_1-p_2}{p_1-p_2} \hat{\xi}_2 \, ,
\eez
and thus
\bez
    (u_{1,\mathrm{out}}, u_{2,\mathrm{out}}) = (u_{1,\mathrm{in}}, u_{2,\mathrm{in}}) 
     \left( \begin{array}{cc} \frac{p_1-p_2}{p_1-q_2} & \frac{p_1-q_1}{p_1-q_2} \\ \frac{p_2-q_2}{p_1-q_2} & \frac{q_1-q_2}{p_1-q_2}
     \end{array} \right) \, .
\eez
Generalizing the matrix that appears on the right hand side to
\be
  R(p_i,q_i;p_j,q_j) = \left( \begin{array}{cc} \frac{p_i-p_j}{p_i-q_j} & \frac{p_i-q_i}{p_i-q_j} \\ \frac{p_j-q_j}{p_i-q_j} & \frac{q_i-q_j}{p_i-q_j}
     \end{array} \right) \, ,    \label{Rmatrix}
\ee
and letting this act from the right on the $i$-th and $j$-th slot of a three-fold direct sum, 
the Yang-Baxter equation holds. 
This can be checked directly or inferred from a 3-soliton solution, see Section~\ref{sec:3solitons}.

\begin{remark}
The reduction to vector KdV$_K$ via $q_i = - p_i$ (see Remark~\ref{rem:KdV_reduction}) leads to
\be
   R(p_i,p_j) = \left( \begin{array}{cc} 
       \frac{p_i - p_j}{p_i+p_j} & \frac{2 p_i}{p_i+p_j} \\
       \frac{2 p_j}{p_i+p_j} & \frac{p_j - p_i}{p_i+p_j}
       \end{array} \right) \, .     \label{R_KdV}
\ee     
This rules the evolution of initial polarizations (at $t \ll 0$) step by step along the tropical limit graph 
in two-dimensional space-time. 
The $R$-matrix (\ref{R_KdV}) also describes the elastic collision of non-relativistic particles with masses 
$p_i$ in one dimension, see \cite{Koul17}. 
\hfill $\square$
\end{remark}

\subsection{Pure row vector 2-soliton solutions}
Now we set $m=1$. Then the $\xi_i$ are scalars and drop out of the relevant formulas. Introducing
\bez
     \hat{\eta}_i = \frac{\eta_i}{\eta_i K} \, ,
\eez
we have
\bez
    u_{1,\mathrm{in}} = \frac{q_1-p_2}{q_1-q_2} \hat{\eta}_1 + \frac{p_2-q_2}{q_1-q_2} \hat{\eta}_2 \, , \quad
    u_{2,\mathrm{in}} = \hat{\eta}_2 \, , \quad
    u_{1,\mathrm{out}} = \hat{\eta}_1 \, , \quad
    u_{2,\mathrm{out}} = \frac{q_1-p_1}{q_1-q_2} \hat{\eta}_1 + \frac{p_1-q_2}{q_1-q_2} \hat{\eta}_2 \, ,
\eez
so that
\bez
    \left( \begin{array}{c} u_{1,\mathrm{out}} \\ u_{2,\mathrm{out}} \end{array} \right)
  =  \left( \begin{array}{cc} \frac{q_2-q_1}{p_2-q_1} & \frac{p_2-q_2}{p_2-q_1} \\ \frac{p_1-q_1}{p_2-q_1} & \frac{p_2-p_1}{p_2-q_1}
     \end{array} \right)
     \left( \begin{array}{c} u_{1,\mathrm{in}} \\ u_{2,\mathrm{in}} \end{array} \right) \, ,
\eez
which determines a Yang-Baxter map. Let
\bez
    \tilde{R}(p_i,q_i;p_j,q_j) := \left( \begin{array}{cc} 
       \frac{q_j - q_i}{p_j - q_i} & \frac{p_j - q_j}{p_j - q_i} \\
       \frac{p_i - q_i}{p_j - q_i} & \frac{p_j - p_i}{p_j - q_i}
       \end{array} \right)  
\eez
act on the $i$-th and $j$-th slot of a direct sum. Then the Yang-Baxter equation holds. 
We note that $\tilde{R}(p_i,q_i;p_j,q_j) = R(q_i,p_i;q_j,p_j)^\intercal$.

\section{Pure 3-soliton solutions}
\label{sec:3solitons}
For $N=3$ we find 
\be
   \tau &=& \kappa_{11} \kappa_{22} \kappa_{33} \, \beta \, e^{\vartheta_{111}} 
          + \kappa_{11} \kappa_{22} \, \alpha_{12} \, e^{\vartheta_{112}}
          + \kappa_{11} \kappa_{33} \, \alpha_{13} \, e^{\vartheta_{121}}   \nonumber \\
        &&  + \kappa_{22} \kappa_{33} \, \alpha_{23} \, e^{\vartheta_{211}}
          + \kappa_{11} \, e^{\vartheta_{122}}
          + \kappa_{22} \, e^{\vartheta_{212}}
          + \kappa_{33} \, e^{\vartheta_{221}}
          + e^{\vartheta_{222}} \, ,    \label{tau_3s}
\ee
where again $\kappa_{ij} = \eta_i K \xi_j$, and
\bez
    \alpha_{ij} &=& 1- \frac{(p_i-q_i)(p_j-q_j)}{(p_i-q_j)(p_j-q_i)} \, 
                   \frac{\kappa_{ij} \kappa_{ji}}{\kappa_{ii} \kappa_{jj}} \, , \\
    \beta &=& -2 + \alpha_{12} + \alpha_{13} + \alpha_{23} 
       + \frac{(p_1-q_1)(p_2-q_2)(p_3-q_3)}{(p_1-q_3)(p_2-q_1)(p_3-q_2)} \, 
                   \frac{\kappa_{12} \kappa_{23} \kappa_{31}}{\kappa_{11} \kappa_{22} \kappa_{33}} \\
     && + \frac{(p_1-q_1)(p_2-q_2)(p_3-q_3)}{(p_1-q_2)(p_2-q_3)(p_3-q_1)} \, 
                   \frac{\kappa_{13} \kappa_{21} \kappa_{32}}{\kappa_{11} \kappa_{22} \kappa_{33}} \, .
\eez
 Furthermore, we obtain
\be
   F &=& \Big( (p_1-q_1) \, \alpha_{23} \, \kappa_{22} \kappa_{33} \, \xi_1 \otimes \eta_1
       + (p_2-q_2) \, \alpha_{13} \, \kappa_{11} \kappa_{33} \, \xi_2 \otimes \eta_2  \nonumber \\
     &&  + (p_3-q_3) \, \alpha_{12} \, \kappa_{11} \kappa_{22} \, \xi_3 \otimes \eta_3
       + \frac{(p_1-q_1)(p_2-q_2)}{q_1-p_2} \, \alpha_{312} \, \kappa_{12} \kappa_{33} \, \xi_1 \otimes \eta_2 \nonumber \\
     &&  + \frac{(p_1-q_1)(p_3-q_3)}{q_1-p_3} \, \alpha_{213} \, \kappa_{13} \kappa_{22} \, \xi_1 \otimes \eta_3
       + \frac{(p_2-q_2)(p_3-q_3)}{q_2-p_3} \, \alpha_{123} \, \kappa_{11} \kappa_{23} \, \xi_2 \otimes \eta_3 \nonumber \\
     &&  + \frac{(p_1-q_1)(q_2-p_2)}{p_1-q_2} \, \alpha_{321} \, \kappa_{21} \kappa_{33} \, \xi_2 \otimes \eta_1
       + \frac{(p_1-q_1)(q_3-p_3)}{p_1-q_3} \, \alpha_{231} \, \kappa_{22} \kappa_{31} \, \xi_3 \otimes \eta_1 \nonumber \\
     &&  + \frac{(p_2-q_2)(q_3-p_3)}{p_2-q_3} \, \alpha_{132} \, \kappa_{11} \kappa_{32} \, \xi_3 \otimes \eta_2
       \Big) \, e^{\vartheta_{111}}  \nonumber \\     
     &&  + \Big( (p_1-q_1) \, \kappa_{22} \, \xi_1 \otimes \eta_1
         + (p_2-q_2) \, \kappa_{11} \, \xi_2 \otimes \eta_2
         + \frac{(p_1-q_1)(p_2-q_2)}{q_1-p_2} \kappa_{12} \, \xi_1 \otimes \eta_2 \nonumber \\
     &&  - \frac{(p_1-q_1)(p_2-q_2)}{p_1-q_2} \kappa_{21} \, \xi_2 \otimes \eta_1
         \Big) \, e^{\vartheta_{112}}  \nonumber  \\
     &&  + \Big( (p_1-q_1) \, \kappa_{33} \, \xi_1 \otimes \eta_1
         + (p_3-q_3) \, \kappa_{11} \, \xi_3 \otimes \eta_3
         + \frac{(p_1-q_1)(p_3-q_3)}{q_1-p_3} \kappa_{13} \, \xi_1 \otimes \eta_3 \nonumber \\
     &&  - \frac{(p_1-q_1)(p_3-q_3)}{p_1-q_3} \kappa_{31} \, \xi_3 \otimes \eta_1
         \Big) \, e^{\vartheta_{121}} \nonumber \\
     &&  + \Big( (p_2-q_2) \, \kappa_{33} \, \xi_2 \otimes \eta_2
         + (p_3-q_3) \, \kappa_{22} \, \xi_3 \otimes \eta_3
         + \frac{(p_2-q_2)(p_3-q_3)}{q_2-p_3} \kappa_{23} \, \xi_2 \otimes \eta_3 \nonumber \\
     &&  - \frac{(p_2-q_2)(p_3-q_3)}{p_2-q_3} \kappa_{32} \, \xi_3 \otimes \eta_2
         \Big) \, e^{\vartheta_{211}} \nonumber \\
     &&  + (p_1-q_1) \, \xi_1 \otimes \eta_1 \, e^{\vartheta_{122}}  
       + (p_2-q_2) \, \xi_2 \otimes \eta_2 \, e^{\vartheta_{212}}   
       + (p_3-q_3) \, \xi_3 \otimes \eta_3 \, e^{\vartheta_{221}} \, ,   \label{F_3s}
\ee
where
\bez
    \alpha _{kij} = 1 
     - \frac{(q_i-p_j)(q_k-p_k) \kappa_{ik} \, \kappa_{kj}}{(q_i-p_k)(q_k-p_j) \kappa_{ij} \, \kappa_{kk}} \, .
\eez
Note that $\alpha_{ij} = \alpha_{ijj}$. 
Recall that the coefficient of $e^{\vartheta_{abc}}$ in the expression for $\tau$, respectively $F$, has been named $\mu_{abc}$, 
respectively $M_{abc}$. The tropical value in the region where $\vartheta_{abc}$ dominates all other phases is given by
\bez
     \phi_{abc} = \frac{M_{abc}}{\mu_{abc}} \, .
\eez
The corresponding values can be read off from (\ref{tau_3s}) and (\ref{F_3s}).

\subsection{Pure vector KP 3-soliton solutions}
Now we restrict our considerations to the vector case $n=1$. Using
\bez
    \hat{\xi}_i = \frac{\xi_i}{K \xi_i} \, ,
\eez
we obtain
\bez
  && \hat{u}_{122,222} = \hat{\xi}_1 \, , \quad 
     \hat{u}_{212,222} = \hat{\xi}_2 \, , \quad 
     \hat{u}_{221,222} = \hat{\xi}_3 \, , \\
  && \hat{u}_{111,211} = \frac{(p_1-q_2)(p_1-q_3)}{(p_1-p_2)(p_1-p_3)} \, \hat{\xi}_1 
               - \frac{(p_2-q_2)(p_2-q_3)}{(p_1-p_2)(p_2-p_3)} \, \hat{\xi}_2
               + \frac{(p_3-q_2)(p_3-q_3)}{(p_1-p_3)(p_2-p_3)} \, \hat{\xi}_3 \, , \\
  && \hat{u}_{111,121} = \frac{(p_1-q_1)(p_1-q_3)}{(p_1-p_2)(p_1-p_3)} \, \hat{\xi}_1 
               - \frac{(p_2-q_1)(p_2-q_3)}{(p_1-p_2)(p_2-p_3)} \, \hat{\xi}_2
               + \frac{(p_3-q_1)(p_3-q_3)}{(p_2-p_3)(p_1-p_3)} \, \hat{\xi}_3 \, , \\
  && \hat{u}_{111,112} = \frac{(p_1-q_1)(p_1-q_2)}{(p_1-p_2)(p_1-p_3)} \, \hat{\xi}_1 
               - \frac{(p_2-q_1)(p_2-q_2)}{(p_1-p_2)(p_2-p_3)} \, \hat{\xi}_2
               + \frac{(p_3-q_1)(p_3-q_2)}{(p_1-p_3)(p_2-p_3)} \, \hat{\xi}_3 \, , \\               
  && \hat{u}_{121,221} = \frac{p_1-q_3}{p_1-p_3} \, \hat{\xi}_1 - \frac{p_3-q_3}{p_1-p_3} \, \hat{\xi}_3  \, , \quad 
     \hat{u}_{121,122} = \frac{p_1-q_1}{p_1-p_3} \, \hat{\xi}_1 - \frac{p_3-q_1}{p_1-p_3} \, \hat{\xi}_3  \, , \\   
  && \hat{u}_{112,122} = \frac{p_1-q_1}{p_1-p_2} \, \hat{\xi}_1 - \frac{p_2-q_1}{p_1-p_2} \, \hat{\xi}_2  \, , \quad 
     \hat{u}_{211,221} = \frac{p_2-q_3}{p_2-p_3} \, \hat{\xi}_2 - \frac{p_3-q_3}{p_2-p_3} \, \hat{\xi}_3  \, , \\ 
  && \hat{u}_{211,212} = \frac{p_2-q_2}{p_2-p_3} \, \hat{\xi}_2 - \frac{p_3-q_2}{p_2-p_3} \, \hat{\xi}_3  \, , \quad 
     \hat{u}_{112,212} = \frac{p_1-q_2}{p_1-p_2} \, \hat{\xi}_1 - \frac{p_2-q_2}{p_1-p_2} \, \hat{\xi}_2  \, .
\eez
\begin{figure} 
\begin{center}
\includegraphics[scale=.25]{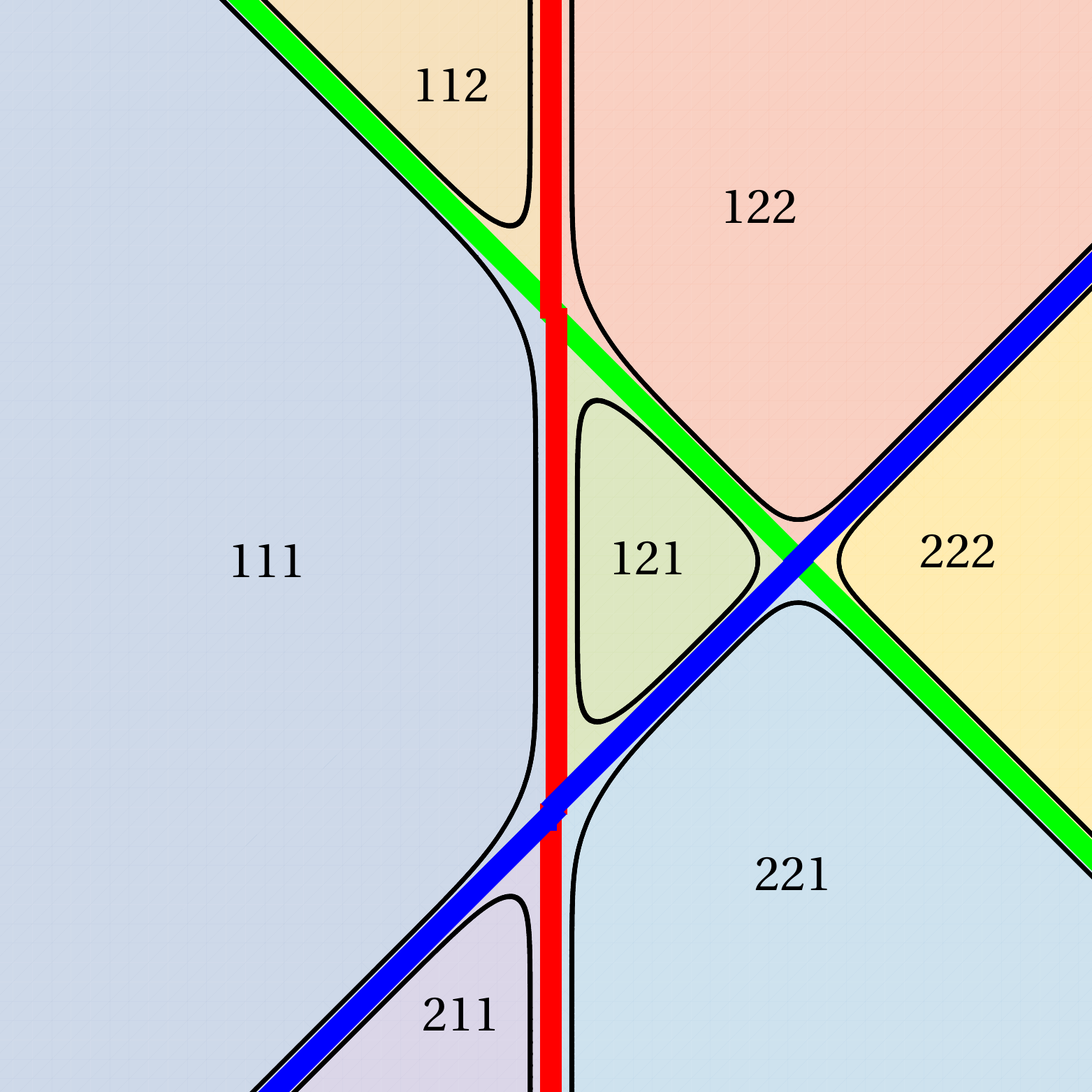} 
\hspace{.5cm}
\includegraphics[scale=.25]{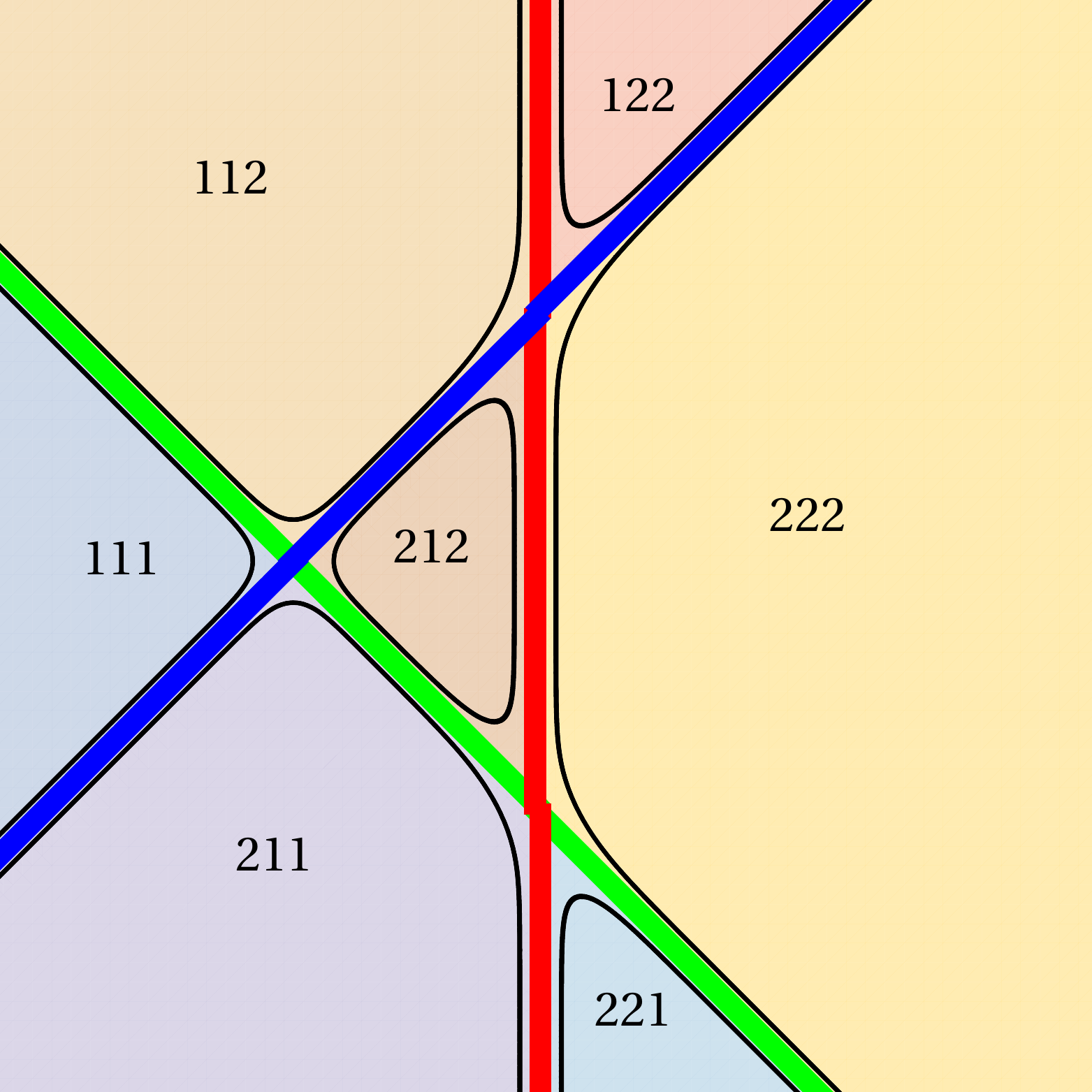} 
\end{center}
\caption{Yang-Baxter relation in terms of vector KP line solitons. These are contour plots in the $xy$-plane 
(horizontal $x$- and vertical $y$-axis) of a 3-soliton solution at negative, respectively positive $t$. 
A number $abc$ indicates the respective dominating phase region.
\label{fig:YB} }
\end{figure}
The contour plots in Fig.~\ref{fig:YB} show the structure at fixed $t$ with $t < 0$ and $t > 0$, respectively.  
The lines extending to the bottom are numbered by $1,2,3$ from left to right (displayed as blue, red, green, respectively). 
Thinking of three particles undergoing a scattering process in $y$-direction, they carry polarizations that change 
at ``crossing points''. As $y$ increases we have
\bez
 &&  (\hat{u}_{111,211},\hat{u}_{211,221},\hat{u}_{221,222}) \stackrel{12}{\mapsto}
   (\hat{u}_{121,221},\hat{u}_{111,121},\hat{u}_{221,222}) \stackrel{13}{\mapsto}
   (\hat{u}_{122,222},\hat{u}_{111,121},\hat{u}_{121,122}) \\
 &&  \stackrel{23}{\mapsto}
   (\hat{u}_{122,222},\hat{u}_{112,122},\hat{u}_{111,112}) 
\eez
for $t<0$, where the $i$-th entry contains the polarization of the $i$-th particle, and 
\bez
 &&  (\hat{u}_{111,211},\hat{u}_{211,221},\hat{u}_{221,222}) \stackrel{23}{\mapsto}
   (\hat{u}_{111,211},\hat{u}_{212,222},\hat{u}_{211,212}) \stackrel{13}{\mapsto}
   (\hat{u}_{112,212},\hat{u}_{212,222},\hat{u}_{111,112}) \\
 &&  \stackrel{12}{\mapsto}
   (\hat{u}_{122,222},\hat{u}_{112,122},\hat{u}_{111,112}) 
\eez
for $t>0$. The numbers $ij$ assigned to the steps refer to the ``particles'' involved.  
In both cases we start and end with the same vectors, and this implies the Yang-Baxter equation for the associated 
transformations. 
Let $V_{a_1a_2a_3,b_1b_2b_3}$ be the column vector formed by the coefficients of $\hat{u}_{a_1a_2a_3,b_1b_2b_3}$ 
with respect to $\hat{\xi}_1,\hat{\xi}_2,\hat{\xi}_3$. The following matrices are composed of these  
column vectors,
\bez
  && U_{123} = \left( \begin{array}{ccc} V_{111,211} & V_{211,221} & V_{221,222} \end{array} \right) \, , \quad
     U_{213} = \left( \begin{array}{ccc} V_{121,221} & V_{111,121} & V_{221,222} \end{array} \right) \, , \\
  && U_{231} = \left( \begin{array}{ccc} V_{122,222} & V_{111,121} & V_{121,122} \end{array} \right) \, , \quad
     U_{321} = \left( \begin{array}{ccc} V_{122,222} & V_{112,122} & V_{111,112} \end{array} \right) \, , \\
  && U_{132} = \left( \begin{array}{ccc} V_{111,211} & V_{212,222} & V_{211,212} \end{array} \right) \, , \quad
     U_{312} = \left( \begin{array}{ccc} V_{112,212} & V_{212,222} & V_{111,112} \end{array} \right) \, . 
\eez 
They represent the triplets of polarizations constituting the above chains. Next we define matrices
\bez
    r_{12} = U_{123}^{-1} U_{213} = U_{312}^{-1} U_{321} \, , \quad
    r_{13} = U_{213}^{-1} U_{231} = U_{132}^{-1} U_{312} \, , \quad
    r_{23} = U_{231}^{-1} U_{321} = U_{123}^{-1} U_{132} \, , 
\eez
which turn out to be given in terms of the $R$-matrix (\ref{Rmatrix}).
For example, 
\bez
     r_{13} = \left(
   \begin{array}{ccc}
     \frac{p_1-p_3}{p_1-q_3} & 0 & \frac{p_1-q_1}{p_1-q_3} \\
                  0 & 1 & 0 \\
     \frac{p_3-q_3}{p_1-q_3} & 0 & \frac{q_1-q_3}{p_1-q_3} 
    \end{array}\right) \, .
\eez
The Yang-Baxter equation reads
\bez
    r_{12} \, r_{13} \, r_{23} = r_{23} \, r_{13} \,  r_{12} \, .
\eez
Fig.~\ref{fig:YB_Ku} shows plots of $K u$ for a choice of the parameters.  
\begin{figure} 
\begin{center}
\includegraphics[scale=.25]{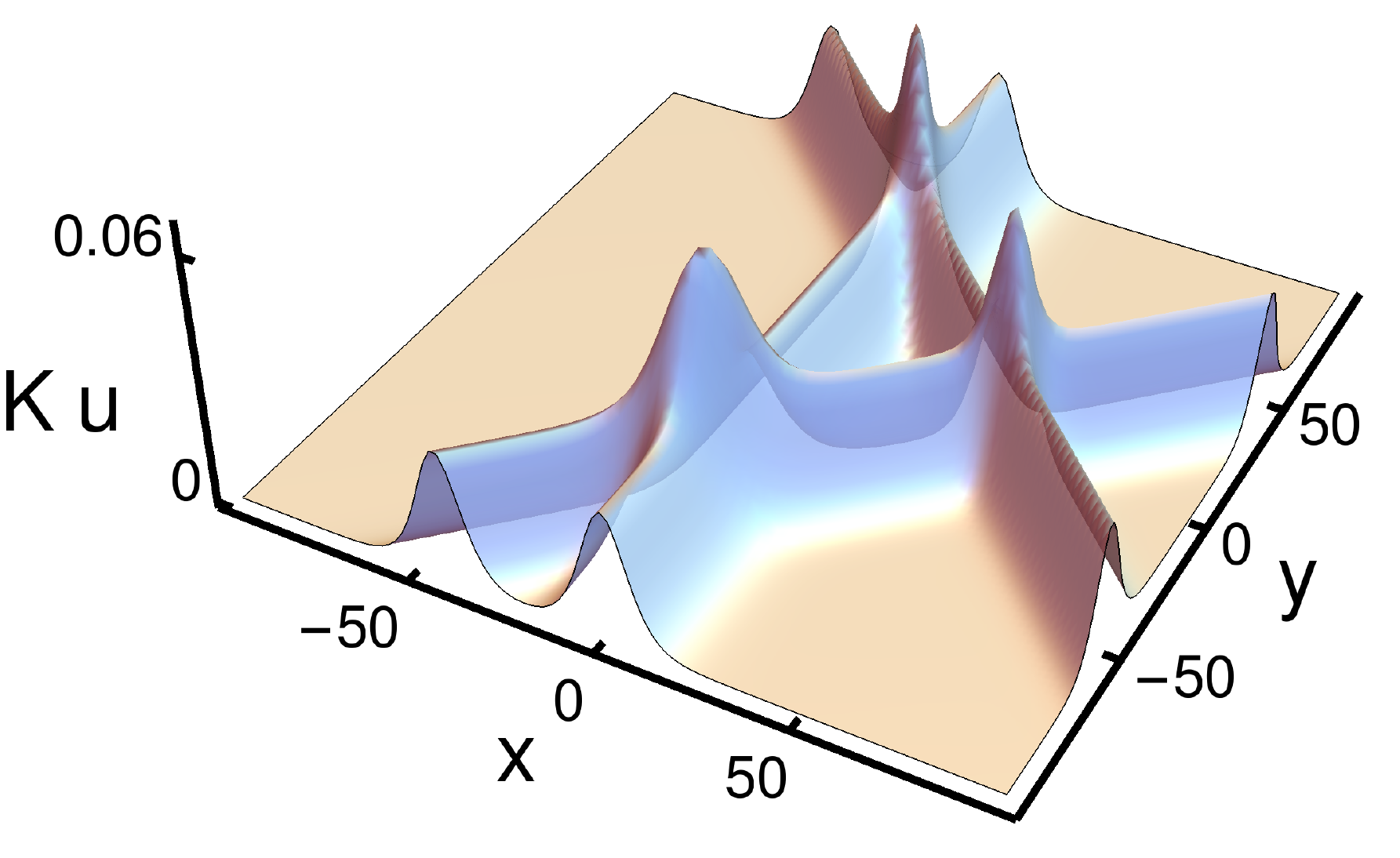} 
\hspace{.5cm}
\includegraphics[scale=.25]{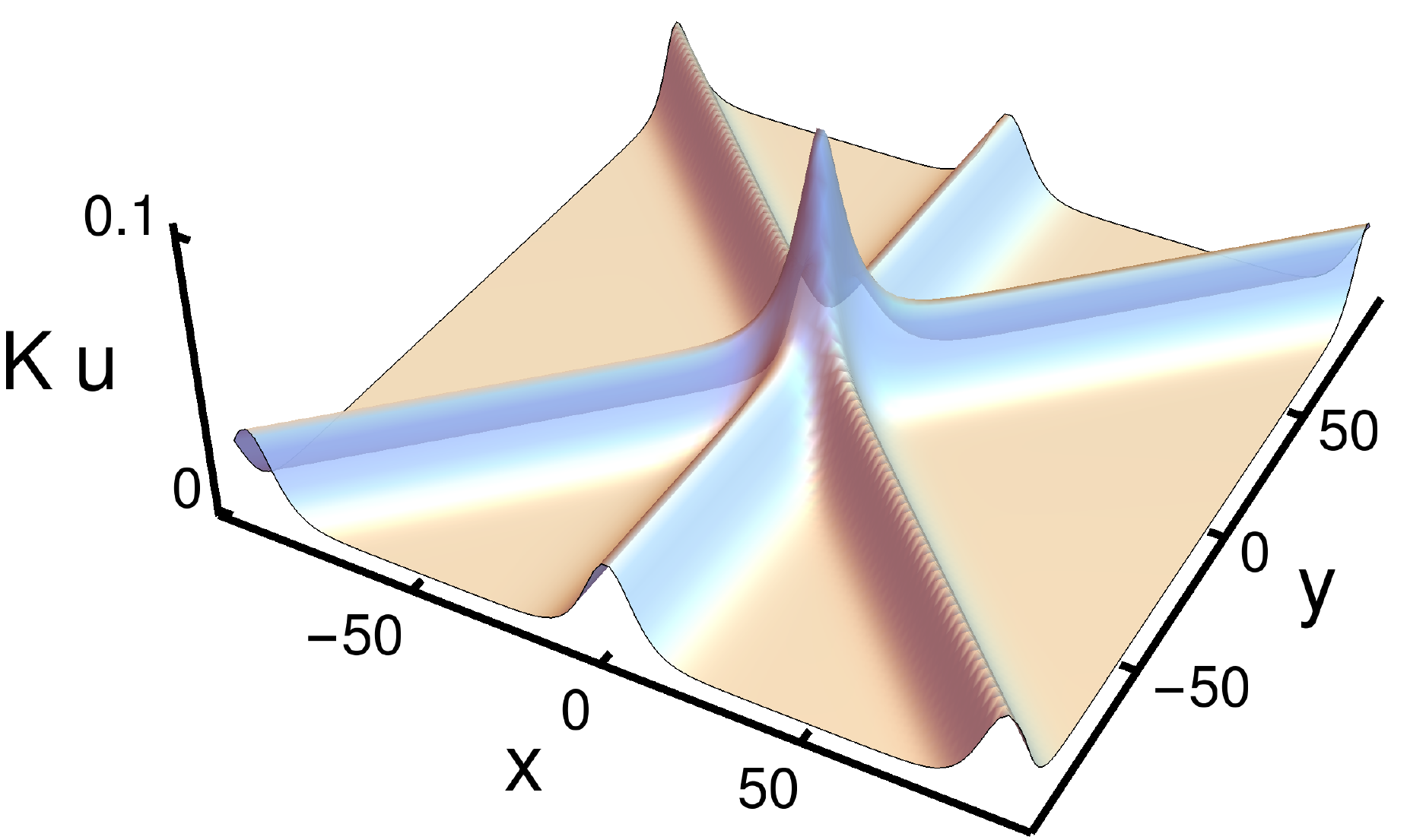} 
\hspace{.5cm}
\includegraphics[scale=.25]{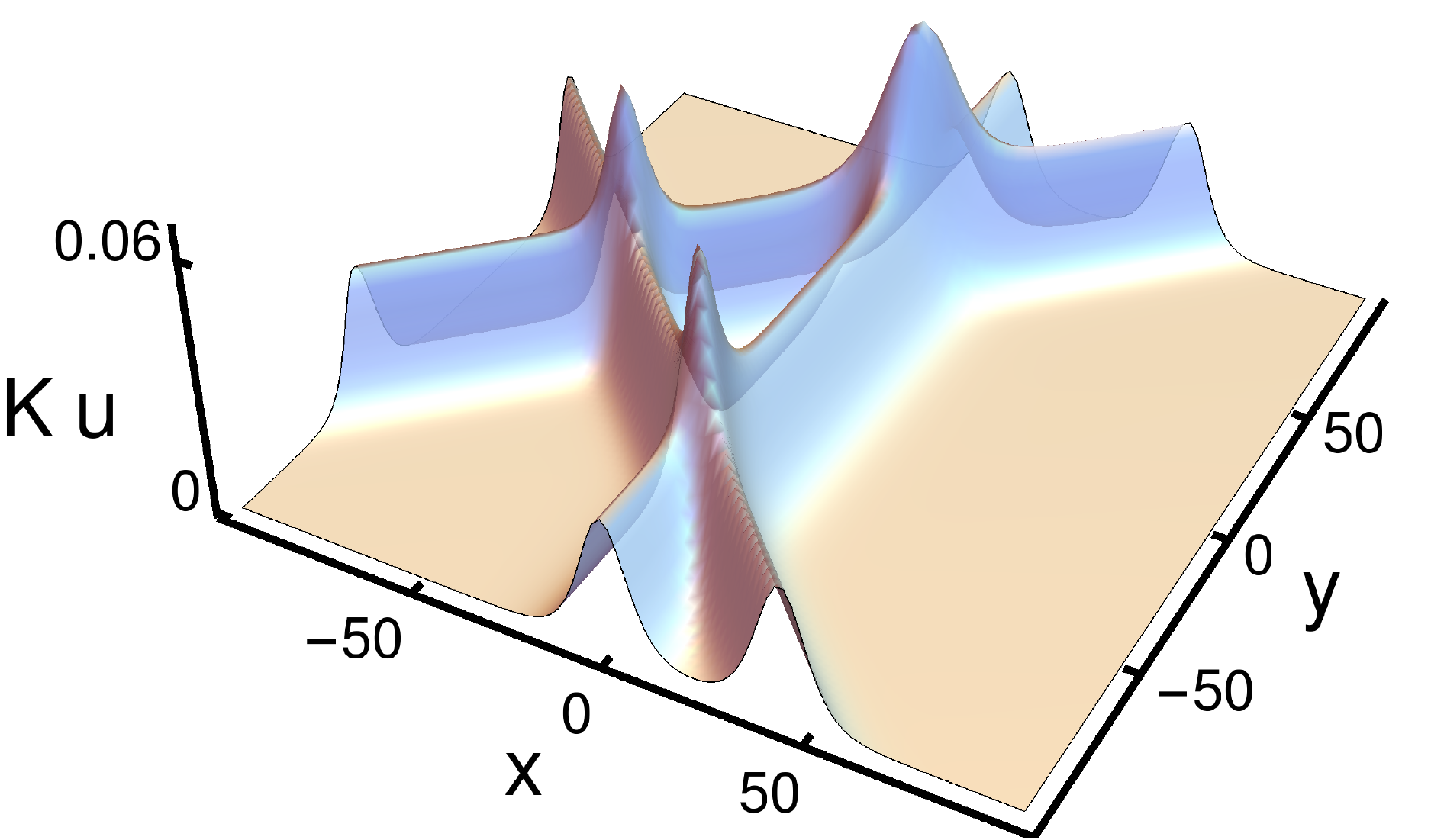} 
\end{center}
\caption{Plots of the scalar $K u$ for a ``Yang-Baxter line soliton configuration'' of a vector KP equation 
at times $t<0$, $t=0$ and $t>0$.
\label{fig:YB_Ku} }
\end{figure}

\subsection{Vector KdV 3-soliton solutions}
\label{subsec:vectorKdV}
We impose the KdV reduction, see Remark~\ref{rem:KdV_reduction}, and replace (\ref{vartheta}) by
\bez
      \vartheta(P) = x \, P + t \, P^3 + s \, P^5 \, .
\eez
The additional last term introduces the next evolution variable $s$ of the KdV hierarchy, also 
see Remark~\ref{rem:hierarchy}.
\begin{figure}[!htbp] 
\begin{center}
\includegraphics[scale=.25]{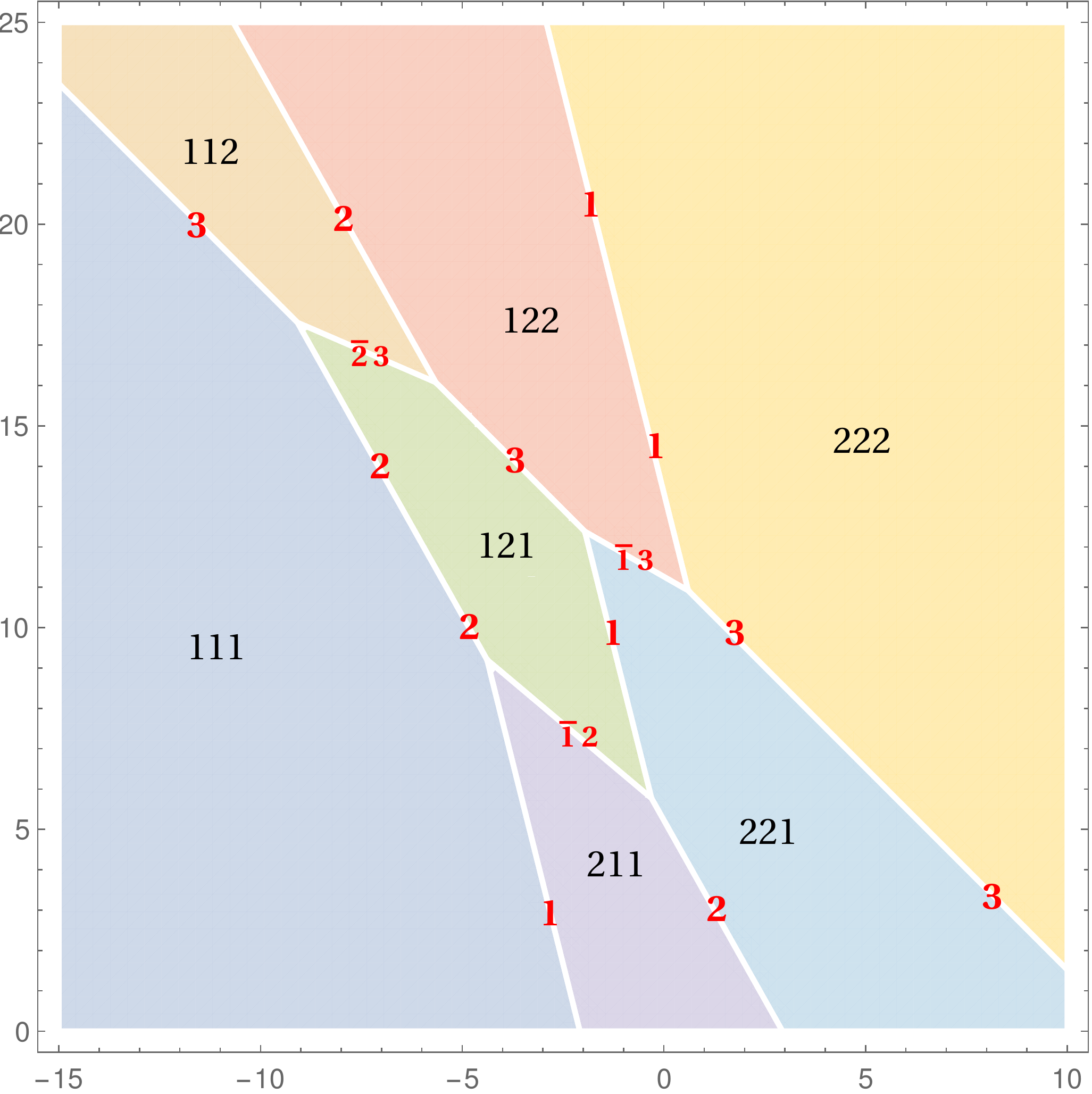} 
\includegraphics[scale=.233]{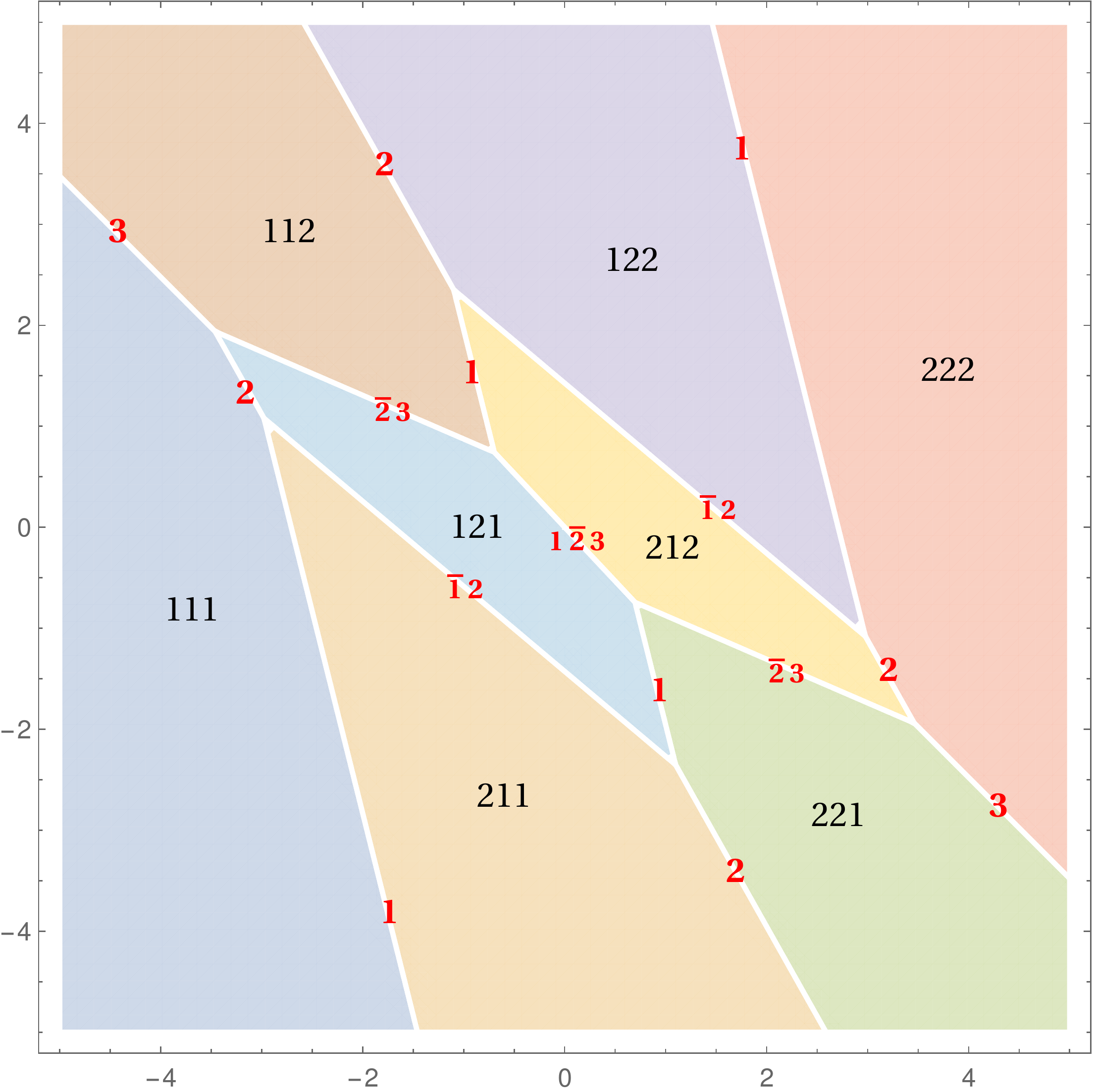} 
\includegraphics[scale=.25]{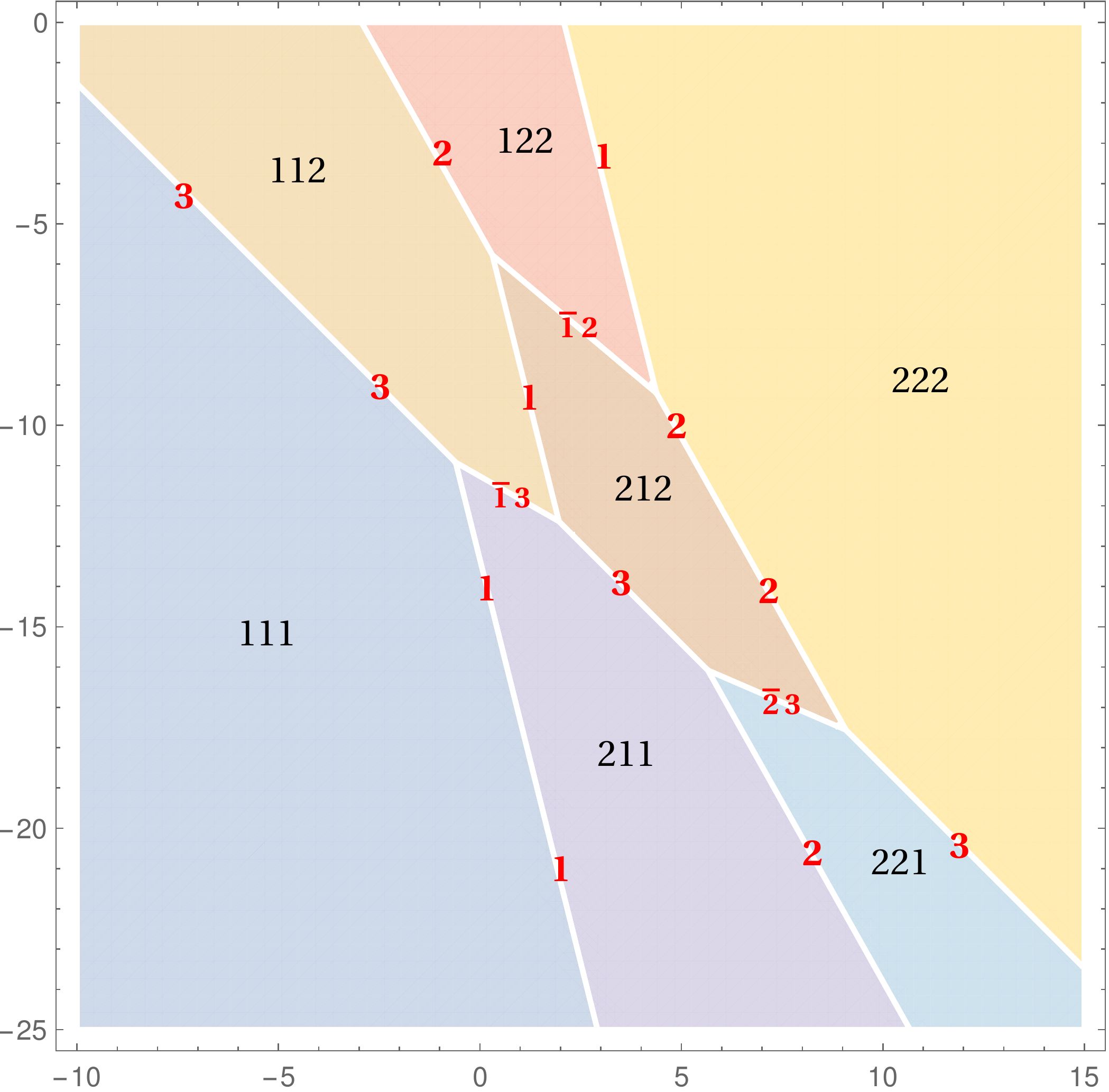} 
\end{center}
\caption{Tropical limit graph and dominating phase regions of a vector KdV solution in two-dimensional space-time 
($x$ horizontal, $t$ vertical), at $s=-10$, $s=0$ and $s=10$. See Section~\ref{subsec:vectorKdV}.
Numbers $1,2,3$ (in red) attached to lines identify appearances of the respective soliton. Here bounded lines 
are formally associated with a pair of a (virtual) anti-soliton, indicated by a bar over the respective number,  
and a (virtual) soliton. At $s=0$, a ``composite'' of three virtual solitons ($1\bar{2}3$) shows up. 
\label{fig:vectorKdV} }
\end{figure}
Let us consider, for simplicity, the KdV$_K$ equation with $m=3$ and $K = (1,1,1)$, and the 
special solution with parameters 
\bez
    \theta_1 = I_3 \, ,  \qquad 
    \chi_1 = (1,-1,1)^\intercal \, .
\eez
The tropical limit graph is displayed in Fig.~\ref{fig:vectorKdV} for $p_1=1/2$, $p_2=3/4$, $p_3=1$, and 
different values of $s$. We have the matrix 
\bez
 U_{123} = \left( \begin{array}{ccc} \hat{u}_{111,211} & \hat{u}_{211,221} & \hat{u}_{221,222} \end{array} \right) 
         = \left(
   \begin{array}{ccc}
   \frac{(p_1+p_2)(p_1+p_3)}{(p_1-p_2)(p_1-p_3)} & 0 & 0  \\
  -\frac{2 p_2 (p_2+p_3)}{(p_1-p_2)(p_2-p_3)} & \frac{p_2+p_3}{p_2-p_3} & 0 \\
   \frac{2 p_3 (p_2+p_3)}{(p_1-p_3)(p_2-p_3)} & -\frac{2 p_3}{p_2-p_3} & 1 
   \end{array} \right) 
\eez
of initial polarizations. The next values $\hat{u}_{abc,def}$ are then obtained by application of the $R$-matrix 
(\ref{R_KdV}) from the right, and so forth, following either the left or the right graph in Fig.~\ref{fig:vectorKdV}
in upwards (i.e., $t$-) direction. Since the initial and the final polarizations are the same, the Yang-Baxter equation holds. 
Here the $R$-matrix describes the time evolution of polarizations in the tropical limit. 

Fig.~\ref{fig:vectorKdV} suggests to think of $R$-matrices as being associated with bounded lines, which may be thought 
of as representing ``virtual solitons''. Interaction of two solitons then means exchange of a virtual soliton.  

At $s=0$, see the plot in the middle of Fig.~\ref{fig:vectorKdV}, something peculiar occurs, namely a sort of 
three-particle interaction. This is a degenerate special case to which the Yang-Baxter description does not apply. 
But this is not relevant for our conclusions.

\section{Tropical limit of pure vector solitons and the $R$-matrix}
\label{sec:vector_case_and_R-matrix}
We set $n=1$ (vector case). The following results describe what happens at a ``crossing'' of two solitons, numbered 
by $i$ and $j$, depicted as a contour plot in Fig.~\ref{fig:crossing}.

\begin{lemma}
\be
   (p_i - p_j) \, \phi_{I_{ij}(1,1)} + (q_i - q_j) \, \phi_{I_{ij}(2,2)}
 = (p_i - q_j) \, \phi_{I_{ij}(1,2)} + (q_i - p_j) \, \phi_{I_{ij}(2,1)} \, .   \label{phi-identity_at_crossing}
\ee
\end{lemma}
\begin{proof}
This is quickly verified for the 2-soliton solution ($N=2$). But at a crossing a general solution $\phi$ is equivalent  
to a 2-soliton solution, since there the four elementary phases $\vartheta(p_{i,a}),\vartheta(p_{j,a})$, $a=1,2$, 
dominate all others, hence the exponential of any other phase vanishes in the tropical limit.   
\end{proof}

\begin{remark}
(\ref{phi-identity_at_crossing}) can be regarded as a vector version of a scalar linear quadrilaterial equation, 
satisfying ``consistency on a cube'', see (15) in \cite{Atki09}. 
Such a linear relation does \emph{not} hold in the matrix case where $m,n>1$. But (\ref{nonlin_phi_eq}), with $\phi_{ab}$ 
replaced by $\phi_{I_{ij}(a,b)}$ (in which case we have $\phi_{I_{ij}(2,2)} \neq 0$, in general) is a nonlinear 
counterpart of (\ref{phi-identity_at_crossing}). The latter can be deduced from it for $n=1$ by using (\ref{tr(K phi_I)}). 
\hfill $\square$
\end{remark}

\begin{figure} 
\begin{center}
\includegraphics[scale=.3]{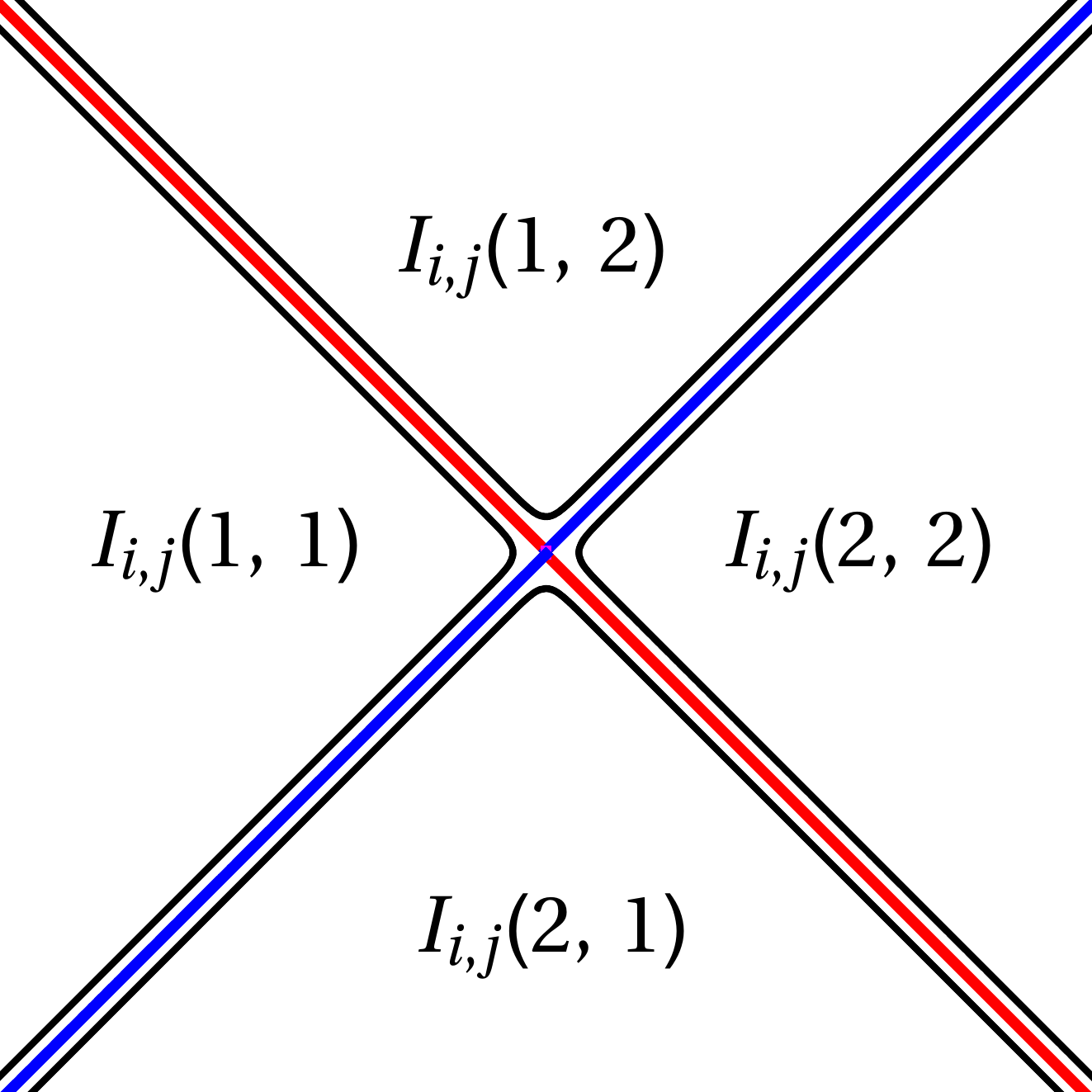} 
\end{center}
\caption{A ``crossing'' of solitons with numbers $i$ and $j$ at fixed time in the $xy$-plane, and the numbers 
of the four phases that are involved.
\label{fig:crossing} }
\end{figure}

\begin{theorem}
\be
    \left( \begin{array}{cc} \hat{u}_{I_{ij}(1,1),I_{ij}(2,1)} & \hat{u}_{I_{ij}(2,1),I_{ij}(2,2)} \end{array} \right) 
    R(p_i,q_i;p_j,q_j) 
    = \left( \begin{array}{cc} \hat{u}_{I_{ij}(1,2),I_{ij}(2,2)} & \hat{u}_{I_{ij}(1,1),I_{ij}(1,2)} \end{array} \right)  
      \label{vector_case_map_R-matrix}
\ee
with
\bez
 R(p_i,q_i;p_j,q_j) = \left(\begin{array}{cc}
 \frac{p_i - p_j}{p_i-q_j} & \frac{p_i-q_i}{p_i-q_j} \\
 \frac{p_j-q_j}{p_i-q_j} & \frac{q_i-q_j}{p_i-q_j} 
 \end{array} \right) \, .
\eez       
\end{theorem}
\begin{proof}
Using (\ref{hatu_IJ}) we can directly verify that the following relations hold as a consequence of (\ref{phi-identity_at_crossing}),
\bez
  &&  \frac{p_i-p_j}{p_i-q_j} \hat{u}_{I_{ij}(1,1),I_{ij}(2,1)} + \frac{p_j-q_j}{p_i-q_j} \hat{u}_{I_{ij}(2,1),I_{ij}(2,2)}
    = \hat{u}_{I_{ij}(1,2),I_{ij}(2,2)} \, , \\
  &&  \frac{p_i-q_i}{p_i-q_j} \hat{u}_{I_{ij}(1,1),I_{ij}(2,1)} + \frac{q_i-q_j}{p_i-q_j} \hat{u}_{I_{ij}(2,1),I_{ij}(2,2)}
    = \hat{u}_{I_{ij}(1,1),I_{ij}(1,2)} \, .
\eez
In matrix form, this is (\ref{vector_case_map_R-matrix}). 
\end{proof}

We already know that $R(p_i,q_i;p_j,q_j)$ satisfies the Yang-Baxter equation.

\section{Construction of pure vector KP soliton solutions from a scalar KP solution and the $R$-matrix}
\label{sec:recon}
Given a $\tau$-function for a pure $N$-soliton solution of the \emph{scalar} KP equation, the Yang-Baxter 
$R$-matrix found above can be used to construct a pure $N$-soliton solution of the vector KP equation.
We will explain this for the case $N=3$.

The $\tau$-function of the pure 3-soliton solution of the scalar KP-II equation is given by
\bez
    \tau = \sum_{a,b,c=1}^2 \Delta_{abc} \, e^{\vartheta_{abc}} \, , \quad
    \Delta_{abc} = (p_{2,b}-p_{1,a})(p_{3,c}-p_{1,a})(p_{3,c}-p_{2,b}) \, ,
\eez
as obtained by the Wronskian method (see, e.g., \cite{Hiro04}). Comparison with (\ref{tau_pure_expansion}) shows that $\mu_I = \Delta_I$. 
\begin{figure} 
\begin{center}
\includegraphics[scale=.28]{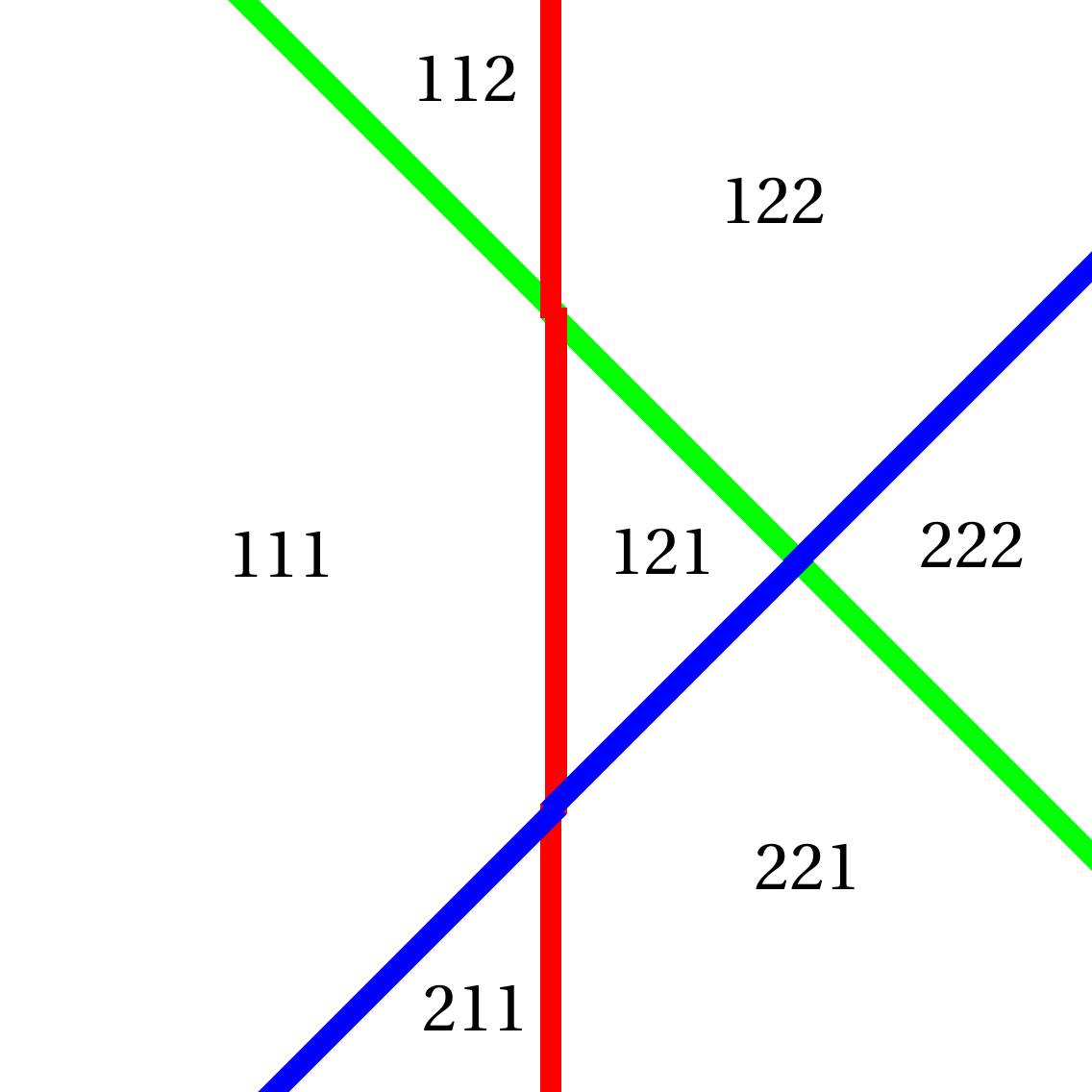} 
\hspace{2cm}
\includegraphics[scale=.28]{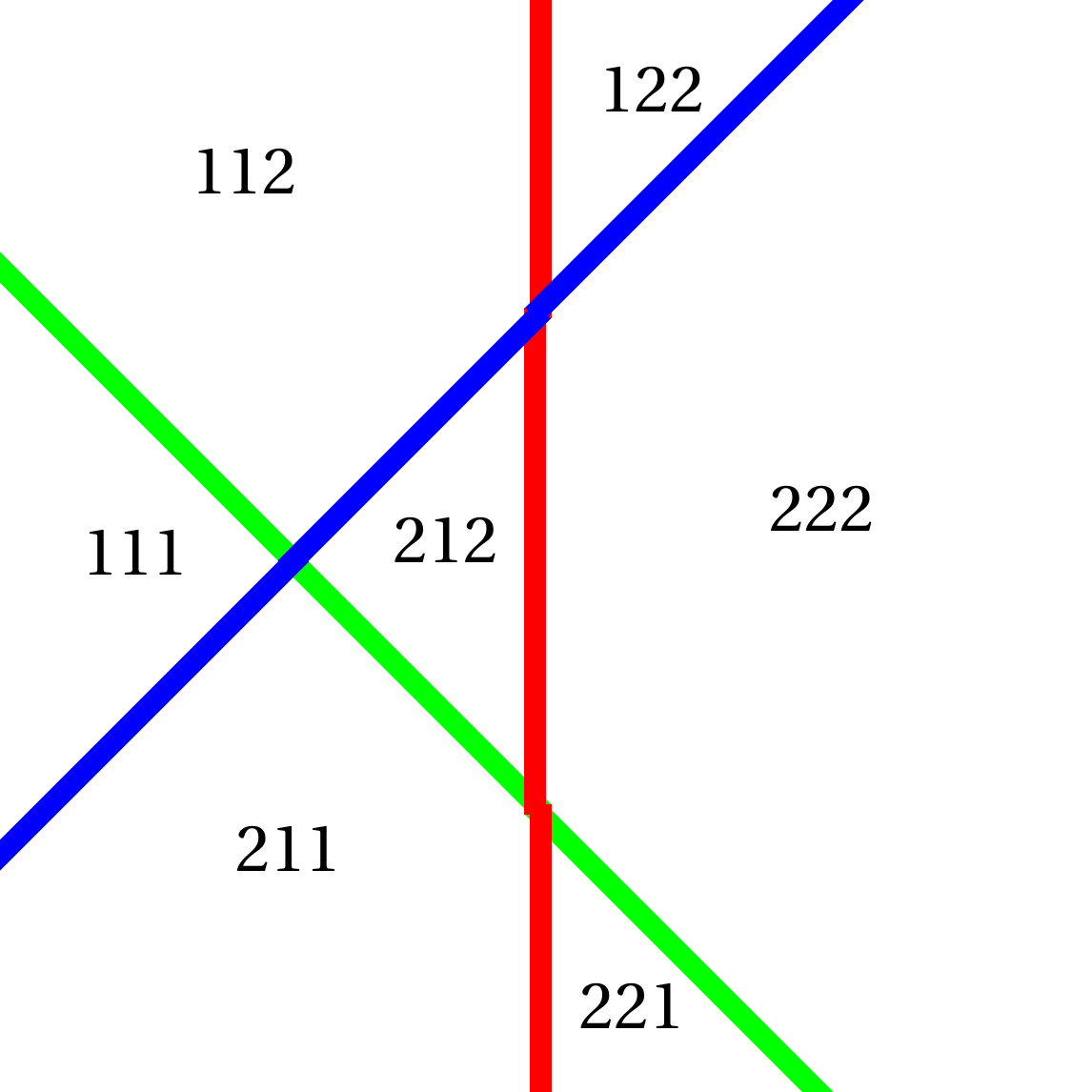} 
\end{center}
\caption{Three-soliton configuration for $t<0$ and $t>0$, respectively, in the $xy$-plane. Numbers specify  
dominating phase regions. 
\label{fig:YB_for_reconstr} }
\end{figure}
Starting at the bottom of both graphs in Fig.~\ref{fig:YB_for_reconstr}, we associate a column vector 
$\xi_i$ with the $i$-th soliton (counted from left to right) and normalize it such that $K \xi_i=1$. 
Accordingly, we set
\bez
    \hat{u}_{111,211} = \xi_1 \, , \quad
    \hat{u}_{211,221} = \xi_2 \, , \quad
    \hat{u}_{221,222} = \xi_3 \, .
\eez
By consecutive application of the $R$-matrix (\ref{Rmatrix}), we find the polarizations on the further 
line segments, proceeding in $y$-direction. There are different ways to proceed, but they are consistent since 
$R$ satisfies the Yang-Baxter equation. For example,
\bez
 \left(\begin{array}{cc}
 \hat{u}_{121,221} & \hat{u}_{111,121} 
       \end{array}\right)
 =\left( \begin{array}{cc}
 \hat{u}_{111,211} & \hat{u}_{211,221} 
         \end{array}\right) \, R(p_1,q_1;p_2,q_2) \, .
\eez
This leads to
\bez
 \hat{u}_{121,221} &=& \frac{p_1-p_2}{p_1-q_2} \xi_1
                     + \frac{p_2-q_2}{p_1-q_2} \xi_2 \, , \qquad
 \hat{u}_{111,121} = \frac{p_1-q_1}{p_1-q_2} \xi_1
                     + \frac{q_1-q_2}{p_1-q_2} \xi_2 \, , \\
 \hat{u}_{212,222} &=& \frac{p_2-p_3}{p_2-q_3} \xi_2
                     + \frac{p_3-q_3}{p_2-q_3} \xi_3 \, , \qquad
 \hat{u}_{211,212} = \frac{p_2-q_2}{p_2-q_3} \xi_2
                     + \frac{q_2-q_3}{p_2-q_3} \xi_3 \, , \\
 \hat{u}_{122,222} &=& \frac{(p_1-p_2)(p_1-p_3)}{(p_1-q_2)(p_1-q_3)} \xi_1
                     + \frac{(p_1-p_3)(p_2-q_2)}{(p_1-q_2)(p_1-q_3)} \xi_2
                     + \frac{p_3-q_3}{p_1-q_3} \xi_3 \, , \\
 \hat{u}_{121,122} &=& \frac{(p_1-p_2)(p_1-q_1)}{(p_1-q_2)(p_1-q_3)} \xi_1
                     + \frac{(p_1-q_1)(p_2-q_2)}{(p_1-q_2)(p_1-q_3)} \xi_2
                     + \frac{q_1-q_3}{p_1-q_3} \xi_3 \, , \\
 \hat{u}_{111,112} &=& \frac{p_1-q_1}{p_1-q_3} \xi_1 + \frac{(p_2-q_2)(q_1-q_3)}{(p_1-q_3)(p_2-q_3)} \xi_2 
                     + \frac{(q_1-q_3)(q_2-q_3)}{(p_1-q_3)(p_2-q_3)} \xi_3 \, , \\
 \hat{u}_{112,212} &=& \frac{p_1-p_3}{p_1-q_3} \xi_1 + \frac{(p_2-q_2)(p_3-q_3)}{(p_1-q_3)(p_2-q_3)} \xi_2 
                     + \frac{(p_3-q_3)(q_2-q_3)}{(p_1-q_3)(p_2-q_3)} \xi_3 \, , \\
 \hat{u}_{112,122} &=& \frac{(p_1-p_3)(p_1-q_1)}{(p_1-q_2)(p_1-q_3)} \xi_1
                     + \frac{(p_2-p_3)(q_1-q_2)(p_1-q_3) 
                     + (p_1-q_1)(p_2-q_2)(p_3-q_3)}{(p_1-q_3)(p_1-q_2)(p_2-q_3)} \xi_2 \\
               &&    + \frac{(p_3-q_3)(q_1-q_3)}{(p_1-q_3)(p_2-q_3)} \xi_3 \, .
\eez
By ``integration'' of (\ref{hatu_IJ}), we find $\phi_{abc}$ up to a single constant. 
Setting $\phi_{222}=0$, we obtain
\bez
  \phi_{111} &=& \xi_1 (p_1-q_1)+\xi_2 (p_2-q_2)+\xi_3 (p_3-q_3) \, , \\
  \phi_{112} &=& \frac{(p_1-p_3) (p_1-q_1)}{p_1-q_3} \xi_1 
                + \Big( \frac{(q_1-q_3) (p_2-q_2) (p_3-q_3)}{(p_1-q_3) (q_3-p_2)} + p_2-q_2 \Big) \xi_2 \\
             &&   + \Big( \frac{(q_1-q_3) (q_2-q_3) (p_3-q_3)}{(p_1-q_3) (q_3-p_2)} + p_3-q_3 \Big) \xi_3  \\
  \phi_{121} &=& \frac{(p_1-p_2) (p_1-q_1)}{p_1-q_2} \xi_1 
                + \frac{(p_1-q_1) (p_2-q_2)}{p_1-q_2} \xi_2 
                + (p_3-q_3) \xi_3 \, , \\
  \phi_{211} &=& (p_2-q_2) \xi_2 + (p_3-q_3) \xi_3 \, ,  \\
  \phi_{122} &=& \frac{(p_1-p_2) (p_1-p_3) (p_1-q_1)}{(p_1-q_2) (p_1-q_3)} \xi_1 
                + \frac{(p_1-p_3) (p_1-q_1) (p_2-q_2)}{(p_1-q_2) (p_1-q_3)} \xi_2 
                + \frac{ (p_1-q_1) (p_3-q_3) }{p_1-q_3} \xi_3 \, , \\
  \phi_{221} &=& (p_3-q_3) \xi_3 \, , \\
  \phi_{212} &=& \frac{(p_2-p_3) (p_2-q_2)}{p_2-q_3} \xi_2 
                + \frac{(p_2-q_2) (p_3-q_3)}{p_2-q_3} \xi_3  \, .
\eez
 From (\ref{phi_I}) we can read off $M_{abc}$ and thus obtain via (\ref{phi=F/tau}) and (\ref{F_pure}) the solution
\bez
     \phi = \frac{1}{\tau} \sum_{a,b,c=1}^2 M_{abc} \, e^{\vartheta_{abc}} 
\eez
of the vector pKP equation.
This procedure can easily be applied to a larger number of solitons.

\section{A solution of the tetrahedron equation}
\label{sec:sol_tetrahedron}
In this section we address again the case of three pure solitons, see Section~\ref{sec:3solitons}. 
Because of the occurrence of phase shifts, in the tropical limit the ``crossing points'' of solitons 
are not really points. Fig.~\ref{fig:YBzoom} shows this for the second interaction (in the vertical $y$-direction) 
in Fig.~\ref{fig:YB}, or Fig.~\ref{fig:YB_for_reconstr}. 
\begin{figure} 
\begin{center}
\includegraphics[scale=.32]{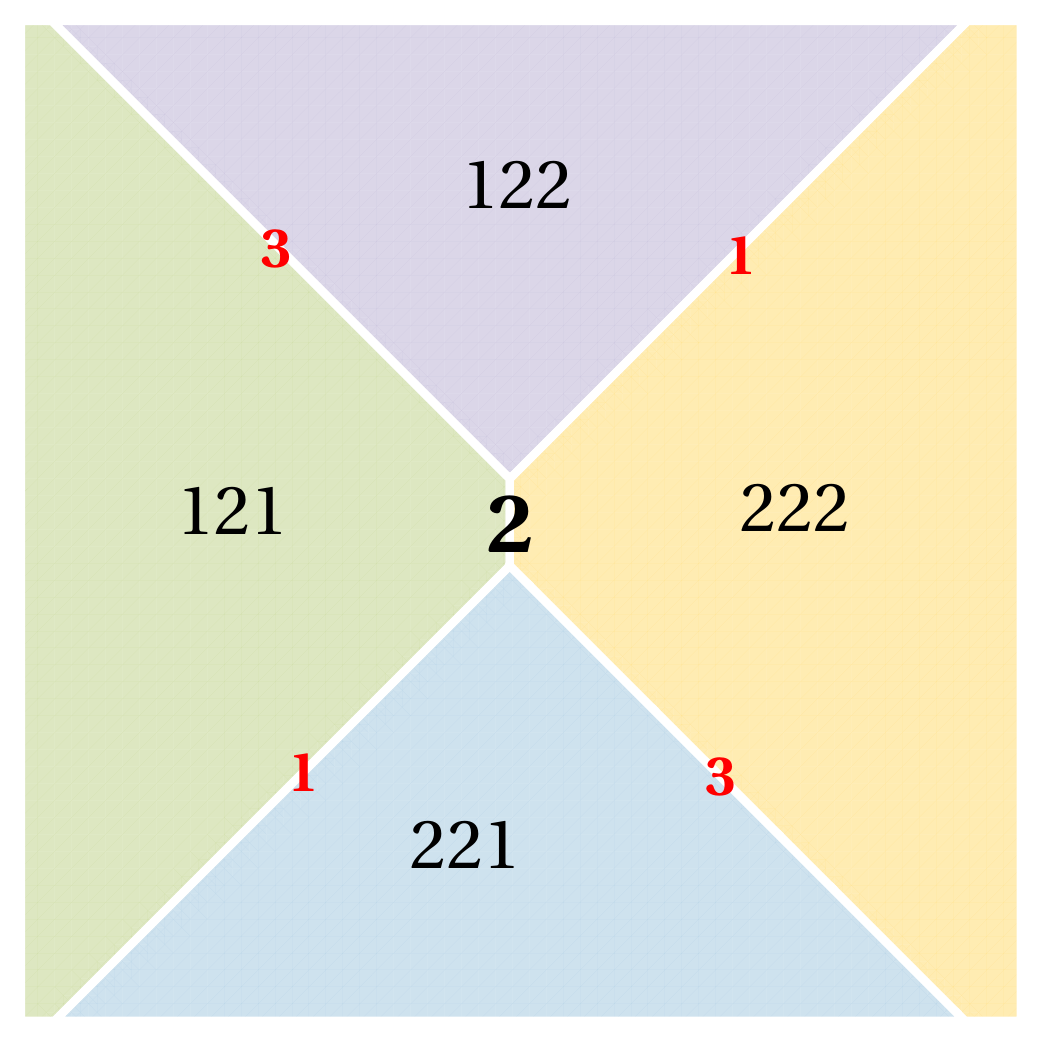} 
\hspace{.5cm}
\includegraphics[scale=.32]{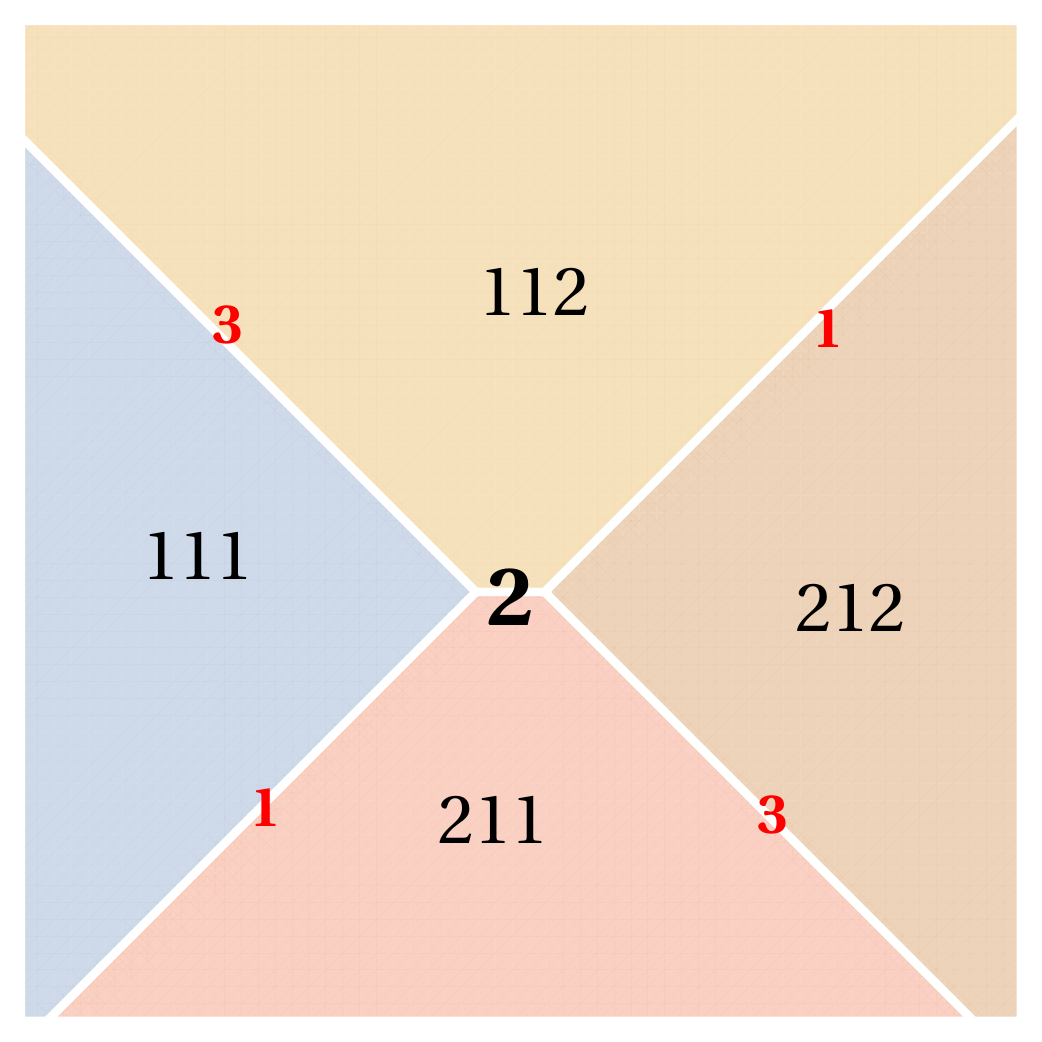} 
\end{center}
\caption{Passing to the tropical limit and zooming into the second interaction ``point'' in the tropical limit of Fig.~\ref{fig:YB}, 
shows the tropical origin of phase shifts. Again, the left figure refers to $t<0$, the right to $t>0$. 
\label{fig:YBzoom} }
\end{figure}
We may think of associating with the additional edge in the left plot ($t<0$) the polarization of the 
boundary between the phase regions numbered by $121 = (1,2,1)$ and $222$, and in the right plot ($t>0$) 
that of the boundary between the phases $112$ and $211$. But this does not lead us to something meaningful. 
Instead, we make an educated guess and associate a mean value of vectors with an additional edge,
\bez
  &&  V_1=\frac{1}{2} \left( V_{111,221}+V_{121,211} \right) \, , \quad
    V_2=\frac{1}{2} \left( V_{121,222}+V_{122,221} \right) \, , \quad
    V_3=\frac{1}{2} \left( V_{111,122}+V_{121,112} \right) \, , 
\eez
and
\bez
 &&  \hat{V}_1=\frac{1}{2} \left( V_{112,222}+V_{122,212} \right) \, , \quad
   \hat{V}_2=\frac{1}{2} \left( V_{111,212}+V_{112,211} \right) \, , \quad
   \hat{V}_3=\frac{1}{2} \left( V_{211,222}+V_{212,221} \right) \, . 
\eez
The inner boundary lines shown in the two plots in Fig.~\ref{fig:YBzoom} are now associated with $V_2$, respectively $\hat{V}_2$. 
In the first plot, the boundary between phase regions $122$ and $221$ is hidden, but there is also 
a polarization associated with it, namely $V_{122,221}$. $V_2$ is the mean value of the latter and $V_{121,222}$, 
which belongs to the visible inner boundary segment. 

Now we have
\bez
     S(i,j,k) = \left(\begin{array}{ccc} V_1 & V_2 & V_3 \end{array}\right)^{-1} 
       \left(\begin{array}{ccc} \hat{V}_1 & \hat{V}_2 & \hat{V}_3 \end{array}\right)\Big|_{p_1 \mapsto p_i, p_2 \mapsto p_j, p_3 \mapsto p_k,q_1 \mapsto q_i, q_2 \mapsto q_j, q_3 \mapsto q_k}  
     = \frac{1}{\delta_{ijk}} \gamma^{ijk} - 1_3 \, , 
\eez
where $\gamma^{ijk} = (\gamma^{ijk}_{rs})$ with
\bez
 && \gamma^{ijk}_{11} = \rho_k \left( \rho_i (p_k+q_k-2 q_i) + \rho_j (p_k+q_k-2 p_j) \right)  , \\
 && \gamma^{ijk}_{22} = \rho_j \left( \rho_i (p_j+q_j-2 q_i) - \rho_k (p_j+q_j-2 q_k) \right)  , \\
 && \gamma^{ijk}_{33} = -\left( \rho_i (\rho_j (p_i+q_i-2 p_j) + \rho_k (p_i+q_i-2 q_k) \right)  , \\
 && \gamma^{ijk}_{12} = -\rho_k \frac{\rho_i^2-\rho_j^2}{\rho_i^2-\rho_k^2} \left( \rho_k^2 - \rho_i (-2 q_i+p_k+q_k) \right)  , \quad
  \gamma^{ijk}_{13} = -\rho_k \frac{\rho_i^2-\rho_j^2}{\rho_j^2-\rho_k^2} \left( \rho_j (p_k+q_k-2 p_j) + \rho_k^2 \right)   , \\
 && \gamma^{ijk}_{21} = -\rho_j \frac{\rho_i^2-\rho_k^2}{\rho_i^2-\rho_j^2} \left( \rho_j^2 - \rho_i (p_j+q_j-2 q_i) \right)   , \quad
   \gamma^{ijk}_{23} = \rho_j \frac{\rho_i^2-\rho_k^2}{\rho_j^2-\rho_k^2} \left( \rho_j^2 - \rho_k (p_j+q_j-2 q_k) \right)    , \\
 && \gamma^{ijk}_{31} = \rho_i \frac{\rho_j^2-\rho_k^2}{\rho_i^2-\rho_j^2} \left( \rho_j (p_i+q_i-2 p_j) + \rho_i^2 \right)    , \quad
   \gamma^{ijk}_{32} = \rho_i \frac{\rho_j^2-\rho_k^2}{\rho_i^2-\rho_k^2} \left( \rho_i^2 - \rho_k (p_i+q_i-2 q_k) \right)    ,
\eez
and 
\bez
   \delta_{ijk} = \left(p_j-q_i\right) \rho_i \rho_j + \left(q_k-q_i\right) \rho_i \rho_k + \left(q_k-p_j\right) \rho_j \rho_k \, .   
\eez
Here we set
\bez
     \rho_i = p_i - q_i \, .
\eez
Let $S_{\alpha\beta\gamma}(i,j,k)$ be the $6 \times 6$ matrix which acts via $S(i,j,k)$ on the positions 
$\alpha,\beta,\gamma$ of a $6$-component column vector, and as the identity on the others. Let  
\bez
    \mathcal{R}_{123} = S_{123}(1,2,3) \, , \quad
    \mathcal{R}_{145} = S_{145}(1,2,4) \, , \quad
    \mathcal{R}_{246} = S_{246}(1,3,4) \, , \quad
    \mathcal{R}_{356} = S_{356}(2,3,4) \, .
\eez
We verified that this constitutes a (to our knowledge new) solution of the \emph{tetrahedron (Zamolodchikov) equation}
(see e.g. \cite{DMH15} and references cited there)
\be
     \mathcal{R}_{123} \, \mathcal{R}_{145} \, \mathcal{R}_{246} \, \mathcal{R}_{356} 
   = \mathcal{R}_{356} \, \mathcal{R}_{246} \, \mathcal{R}_{145} \, \mathcal{R}_{123} \, .  \label{tetrahedron_eq}
\ee
An explanation for the choices of numbers in the above definition of $\mathcal{R}_{\alpha\beta\gamma}$ can be found 
in Fig.~\ref{fig:tetrahedron}. Also see \cite{KKS98}.

\begin{figure} 
\begin{center}
\includegraphics[scale=.28]{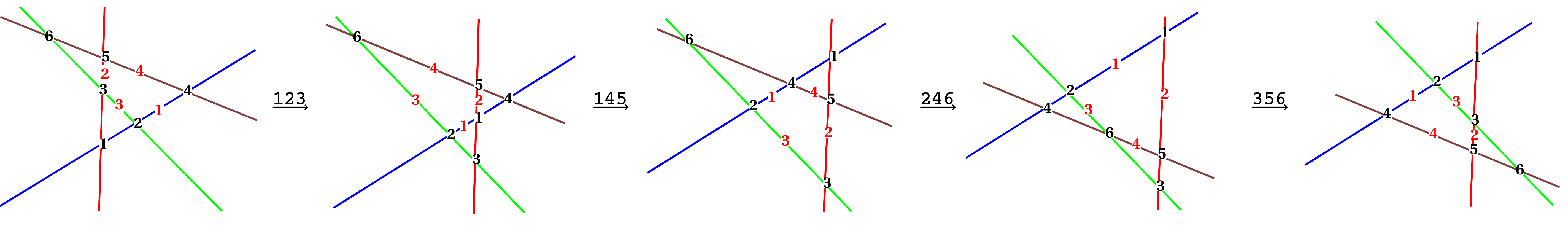} \\
\includegraphics[scale=.29]{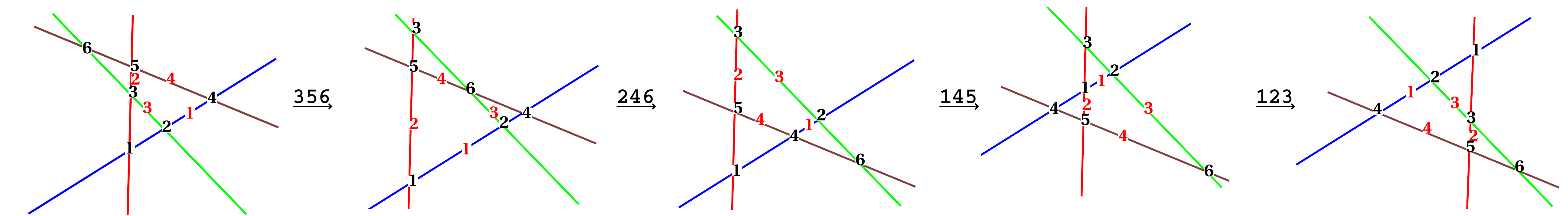} 
\end{center}
\caption{KP line solitons remain parallel while moving. The two chains of line configurations 
(here shifts are disregarded) show the two different ways in which a 4-soliton solution can evolve 
from the same initial to the same final configuration. 
(The two chains represent the higher Bruhat order $B(4,2)$, cf. \cite{DMH15}.) 
This implies the tetrahedron equation.  
Here (red) numbers attached to lines enumerate the four solitons. Also the crossings of pairs of them are 
enumerated (by numbers in black). 
Time evolution proceeds by inversion of triangles. In the first chain, the first step is the inversion of 
the triangle formed by the crossings numbered $1,2,3$. The second step is the inversion of the triangle 
formed by the crossings $1,4,5$. The latter involves the solitons with numbers $1,2,4$.    
\label{fig:tetrahedron} }
\end{figure}

\begin{remark}
The KdV$_K$ reduction $q_i = -p_i$ yields 
\bez
S_{\mathrm{KdV}}(i,j,k) = \left(
 \begin{array}{ccc}
 -\frac{(p_i+p_k) (p_j-p_k)}{(p_i-p_k) (p_j+p_k)} & \frac{2 (p_i-p_j) p_k}{(p_i-p_k) (p_j+p_k)} & \frac{2 (p_i-p_j) p_k}{(p_i-p_k) (p_j+p_k)} \\
 \frac{2 p_j (p_i+p_k)}{(p_i+p_j) (p_j+p_k)} & \frac{(p_i-p_j) (p_j-p_k)}{(p_i+p_j) (p_j+p_k)} & \frac{2 p_j (p_i+p_k)}{(p_i+p_j) (p_j+p_k)} \\
 \frac{2 p_i (p_j-p_k)}{(p_i+p_j) (p_i-p_k)} & \frac{2 p_i (p_j-p_k)}{(p_i+p_j) (p_i-p_k)} & -\frac{(p_i-p_j) (p_i+p_k)}{(p_i+p_j) (p_i-p_k)} \\
 \end{array} \right) \, ,
\eez
which determines a simpler solution of the tetrahedron equation. 
\hfill $\square$
\end{remark}

\section{A generalization of the vector KP $R$-matrix and a solution of the functional tetrahedron equation} 
\label{sec:genR&tetra}
The vector KdV $R$-matrix (\ref{R_KdV}) is obtained from the one-parameter $R$-matrix (see, e.g., \cite{Serg98,KKS98} for a 
similar $R$-matrix)
\bez
     R(x) =  \left( \begin{array}{cc} 
       x & 1+x \\
       1-x & -x
       \end{array} \right) \, ,
\eez
by setting $x = (p_1-p_2)/(p_1+p_2)$. 
The local Yang-Baxter equation
\bez
    R_{12}(x) \, R_{13}(y) \, R_{23}(z) = R_{23}(Z) \, R_{13}(Y) \, R_{12}(X) \, ,
\eez
where indices $\alpha \beta$ indicate on which components of a three-fold direct sum $R$ acts, 
determines the map $(x,y,z) \mapsto (X,Y,Z)$ given by 
\bez
    X = \frac{x y}{x+z - x y z} \, , \quad
    Y = x + z - x y z \, , \quad
    Z = \frac{y z}{x+z - x y z} \, .
\eez
A similar map appeared in \cite{Serg98,KKS98}. 
A general argument (cf., e.g., \cite{DMH15} and references cited there) implies that 
\bez
     \mathcal{R}(x,y,z) := (X,Y,Z) 
\eez
solves the (functional) tetrahedron equation (\ref{tetrahedron_eq}), where a ``product'' of $\mathcal{R}$'s now has 
to be interpreted as composition of maps. This tetrahedron map is involutive.  
Setting $x = (p_1-p_2)/(p_1+p_2)$, $y= (p_1-p_3)/(p_1+p_3)$ and $z= (p_2-p_3)/(p_2+p_3)$, it becomes the identity. 

Correspondingly, the vector KP $R$-matrix (\ref{Rmatrix}) is obtained from the more general two-parameter $R$-matrix
\bez
     R(x,y) =  \left( \begin{array}{cc} 
       x & y \\
       1-x & 1-y
       \end{array} \right) \, ,
\eez
by setting $x = (p_1-p_2)/(p_1-q_2)$ and $y = (p_1-q_1)/(p_1-q_2)$. 
The local Yang-Baxter equation
\bez
    R_{12}(x,y) \, R_{13}(z,u) \, R_{23}(v,w) = R_{23}(V,W) \, R_{13}(Z,U) \, R_{12}(X,Y) \, ,
\eez
determines the map $(x,y;z,u;v,w) \mapsto (X,Y;Z,U;V,W)$, where
\bez
 && X = z \, C \, , \quad
    Y = (z - \frac{A}{x}) \, C \, , \quad
    Z = \frac{x}{C} \, , \\
 && U = 1 - B  \, , \quad
    V =  \frac{v z \, (x-y)}{A} \, , \quad
    W = 1 - \frac{(1-u) (1-w)}{B}  \, ,
\eez
with 
\bez
 && A = u v x-u x-v y+x z   \, , \quad
    B = u w x-u x-w y+1 \, , \\
 && C = \frac{A B-A (1-u) (1-w) x-B v \, (x-y)}{A B-A (1-u) (1-w)-B v z \, (x-y)} \, .
\eez
Then 
\bez
    \mathcal{R}(x,y;z,u;v,w) := (X,Y;Z,U;V,W) 
\eez
solves the functional (i.e., set-theoretical) tetrahedron equation.

\section{Conclusions}
\label{sec:conclusions}
In this work we explored ``pure'' soliton solutions of matrix KP equations in a tropical limit. 
In case of the reduction to matrix KdV, this consists of a planar graph in (two-dimensional) space-time, 
with polarizations assigned to its edges. 
Given initial polarizations, the evolution of them along the graph is ruled by a Yang-Baxter map. 
For the vector KdV equation, this is a \emph{linear} map, hence an $R$ matrix. 
The classical scattering process of matrix KdV solitons resembles in the tropical limit the scattering 
of point particles in a 2-dimensional integrable quantum field theory, which is characterized by a 
scattering matrix that solves the (quantum) Yang-Baxter equation. 

We have shown that all this holds more generally for KP$_K$, where the tropical limit at a fixed 
time $t$ is given by a graph in the $xy$-plane, with polarizations attached to the soliton lines. 
Moreover, the vector KP case provides us with a realization of the ``classical straight-string model'' 
considered in \cite{KKS98}. It should be noticed, however, that KP line solitons in the tropical limit 
are not, in general, straight because of the appearance of (phase) shifts.  

As a side product of our explorations of the tropical limit of pure vector KP solitons, we derived 
apparently new solutions of the tetrahedron (Zamolodchikov) equation. Whether these solutions are relevant, e.g.,  
for the construction of solvable models of statistical mechanics in three dimensions, has still to be seen.

Another subclass of soliton solutions of the vector KP equation consists of those, for which the support 
at fixed time is a rooted and generically binary tree in the tropical limit. For the scalar KP equation, 
this has been extensively explored in \cite{DMH11KPT,DMH12KPBT}. Instead of the Yang-Baxter equation, 
the \emph{pentagon equation} (see \cite{DMH15} and references therein) now plays a role in governing corresponding 
vector solitons. This will be treated in a separate work. 

\paragraph{Acknowledgment.} A.D. thanks V. Papageorgiou for a very helpful discussion.

\end{document}